\documentclass[]{article}
\usepackage{shortcuts}
\usepackage{tikz,amsmath,mathtools}
\usetikzlibrary{decorations.pathreplacing}
\usetikzlibrary{arrows.meta}

\numberwithin{equation}{section}
\numberwithin{thm}{section}

\DeclareMathOperator{\Unif}{Unif}
\DeclareMathOperator{\supp}{supp}
\DeclareMathOperator{\block}{block}
\DeclareMathOperator{\rank}{rank}

\newcommand{\KRR}{\mathrm{KRR}}
\newcommand{\KKMC}{\mathrm{KKMC}}

\title{Tight Kernel Query Complexity of Kernel Ridge Regression and Kernel $k$-means Clustering}
\author{
  Manuel Fernández V \\ Carnegie Mellon University \\ \texttt{manuelf@andrew.cmu.edu} \and
  David P. Woodruff \\ Carnegie Mellon University \\ \texttt{dwoodruf@cs.cmu.edu} \and
  Taisuke Yasuda \\ Carnegie Mellon University \\ \texttt{taisukey@andrew.cmu.edu}
}
\hypersetup{pdftitle={Tight Kernel Query Complexity of Kernel Ridge Regression and Kernel \texorpdfstring{$k$}{k}-means Clustering}}

\begin{document}
\maketitle
\begin{abstract}
    We present tight lower bounds on the number of kernel evaluations required to approximately solve kernel ridge regression (KRR) and kernel $k$-means clustering (KKMC) on $n$ input points. For KRR, our bound for relative error approximation to the minimizer of the objective function is $\Omega(nd_{\mathrm{eff}}^\lambda/\varepsilon)$ where $d_{\mathrm{eff}}^\lambda$ is the effective statistical dimension, which is tight up to a $\log(d_{\mathrm{eff}}^\lambda/\varepsilon)$ factor. For KKMC, our bound for finding a $k$-clustering achieving a relative error approximation of the objective function is $\Omega(nk/\varepsilon)$, which is tight up to a $\log(k/\varepsilon)$ factor. Our KRR result resolves a variant of an open question of El Alaoui and Mahoney, asking whether the effective statistical dimension is a lower bound on the sampling complexity or not. Furthermore, for the important practical case when the input is a mixture of Gaussians, we provide a KKMC algorithm which bypasses the above lower bound.
\end{abstract}

\newpage

\section{Introduction}
The \emph{kernel trick} in machine learning is a general technique that takes linear learning algorithms that only depend on the dot products of the data, including linear regression, support vector machines, principal component analysis, and $k$-means clustering, and boosts them to powerful nonlinear algorithms. This is done by replacing the inner product between two data points with their inner product after applying a kernel map, which implicitly maps the points to a higher dimensional space via a non-linear feature map. The simplicity and power of kernel methods has lead to wide adoption across the machine learning community: nowadays, kernel methods are a staple both in theory \cite{friedman2001elements} and in practice \cite{scholkopf2004kernel, zhang2007local}. We refer the reader to \cite{scholkopf2001learning} for more background on kernel methods.

However, one problem with kernel methods is that the computation of the kernel matrix $\bfK$, the matrix containing all pairs of kernel evaluations between $n$ data points, requires $\Omega(n^2)$ time, which is prohibitively expensive for the large-scale data sets encountered in modern data science. To combat this, a large body of literature in the last decade has been devoted to designing faster algorithms that attempt to trade a small amount of accuracy in exchange for speed and memory, based on techniques such as random Fourier features \cite{rahimi2008random}, sampling \cite{bach2013sharp, alaoui2015fast, musco2017recursive, musco2017sublinear}, sketching \cite{yang2017randomized}, and incomplete Cholesky factorization \cite{bach2002kernel, fine2001efficient}. We refer the reader to the exposition of \cite{musco2017recursive} for a more extensive overview of recent literature on the approximation of kernel methods.

\subsection{Previous work on kernel query complexity}
In this work, we consider lower bounds on the query complexity of the kernel matrix. The kernel query complexity is a fundamental information-theoretic parameter of kernel problems and both upper and lower bounds have been studied by a number of works \cite{lin2014sample, cesa2015complexity, musco2017recursive, musco2017sublinear}.

For kernel ridge regression, a lower bound has been shown for additive error approximation of the objective function value in Corollary 8 of \cite{cesa2015complexity}, which is a weaker approximation guarantee than what we study in this work. However, their bound is not known to be tight. Furthermore, the best known upper bounds for kernel ridge regression are in terms of a data-dependent quantity known as the \emph{effective statistical dimension} \cite{alaoui2015fast, musco2017recursive}, on which the \cite{cesa2015complexity} bound does not depend. The question of whether the effective statistical dimension gives a lower bound on the sample complexity has been posed as an open question by El Alaoui and Mahoney \cite{alaoui2015fast}. We will answer this question affirmatively under a slightly different approximation guarantee than they use, which is nevertheless satisfied by known algorithms nearly tightly, for instance by \cite{musco2017recursive}.

Another kernel problem for which lower bounds have been shown is the problem of giving a $(1+\eps)$ relative Frobenius norm error rank $k$ approximation of the kernel matrix, which has a bound of $\Omega(nk/\eps)$ by Theorem 13 of \cite{musco2017sublinear}. For kernel $k$-means clustering, there are no kernel complexity lower bounds to our knowledge.

Similar cost models have also been studied in the context of semisupervised/interactive learning. Intuitively, kernel evaluations are queries that ask for the similarity between two objects, where the notion of similarity in this context is the implicit notion of similarity recognized by humans, i.e.\ the ``crowd kernel''. In such situations, the dominant cost is the number of these queries that must be made to users, making kernel query complexity an important computational parameter. Mazumdar and Saha \cite{mazumdar2017clustering} study the problem of clustering under the setting where the algorithm obtains information by adaptively asking users whether two data points belong to the same cluster or not. In this setting, the dominant cost that is analyzed is the number of same-cluster queries that the algorithm must make, which exactly corresponds to the kernel query complexity of clustering a set of $n$ points drawn from $k$ distinct points with the indicator function kernel and the 0-1 loss (as opposed to $k$-means clustering, which uses the $\ell_2$ loss). In \cite{tamuz2011adaptively}, the authors consider the problem of learning a ``crowd kernel'', where the implicit kernel function is crowdsourced and the cost is measured as the number of queries of the form ``is $a$ more similar to $b$ than $c$?'' rather than queries that directly access the underlying kernel evaluations.

\subsection{Our contributions}
In this work, we resolve the kernel query complexity of kernel ridge regression and kernel $k$-means clustering up to $\log(d_\eff^\lambda/\eps)$ and $\log(k/\eps)$ factors, respectively. Our lower bounds apply even to \emph{adaptive} algorithms, that is, algorithms that are allowed to decide which kernel entries to query based on the results of previous kernel queries. This is a crucial aspect of our contributions, since some of the most efficient algorithms known for kernel ridge regression and kernel $k$-means clustering make use of adaptive queries, most notably through the use of a data-dependent sampling technique known as \emph{ridge leverage score sampling} \cite{musco2017recursive}.

For kernel ridge regression, we present Theorem \ref{thm:krr-hardness-main}, in which we construct a distribution over kernel ridge regression instances such that any randomized algorithm requires $\Omega(nd_\eff^\lambda/\eps)$ adaptive kernel evaluations. This matches the upper bound given in Theorem 15 of \cite{musco2017recursive} up to a $\log(d_\eff^\lambda/\eps)$ factor. Although we present the main ideas of the proof using the kernel as the dot product kernel, our proof in fact applies to any kernel that is of the form $(c_1-c_0)\mathbbm{1}(\bfe_i=\bfe_j)+c_0$ for constants $c_1>c_0$ when restricted to the standard basis vectors (Theorem \ref{thm:krr-hardness-main-indicator}). This includes any kernel that can be written as functions of dot products and Euclidean distances, including the polynomial kernel and the Gaussian kernel. This result resolves a variant of an open question posed by \cite{alaoui2015fast}, which asks whether the effective statistical dimension is a lower bound on the sampling complexity or not. In their paper, they consider the approximation guarantee of a $(1+\eps)$ relative error in the statistical risk, while we consider a $(1+\eps)$ relative error approximation of the minimizer of the KRR objective function. By providing tight bounds on the query complexity in terms of the effective statistical dimension $d_\eff^\lambda$, we definitively establish the fundamental importance of the quantity as a computational parameter, in addition to its established significance as a statistical parameter in the statistics literature \cite{friedman2001elements}. Furthermore, our result also clearly gives a lower bound on the time complexity of kernel ridge regression that matches Theorem 15 of \cite{musco2017recursive} up to a $\tilde O(d_\eff^\lambda/\eps)$ factor for intermediate ranges of $\eps$. This is in contrast to the conditional $\Omega(n^{2-o(1)})$ time complexity lower bound of \cite{backurs2017fine}, which operates in the regime of $\eps = \exp(-\omega(\log^2 n))$ for approximating the argmin of the objective function.

For kernel $k$-means clustering, we present Theorem \ref{thm:main-lower-bound}, which shows a lower bound of $\Omega(nk/\eps)$ for the problem of outputting a clustering which achieves a $(1+\eps)$ relative error value in the objective function. This matches the upper bound given in Theorem 16 of \cite{musco2017recursive} up to a $\log(k/\eps)$ factor. We also note that the problem of outputting a $(1\pm\eps)$ relative error approximation of the optimal cost itself has an $O(nk) + \poly(k,1/\eps,\log n)$ algorithm, and we complement it with a lower bound of $\Omega(nk)$ in Proposition \ref{prop:kkmc-cost-hardness}.

\begin{figure}
  \centering
  \begin{tabular}{|l|l|l|}
    \hline
    Kernel problem & Upper bound & Lower bound \\
    \hline
    KRR & $O\parens*{\tfrac{nd_\eff^\lambda}{\eps}\log\tfrac{d_\eff^\lambda}{\eps}}$ (\cite{musco2017recursive}, Theorem 15) & $\Omega\parens*{\tfrac{nd_\eff^\lambda}{\eps}}$ (this work, Theorem \ref{thm:krr-hardness-main}) \\
    KKMC & $O\parens*{\tfrac{nk}{\eps}\log\tfrac{k}{\eps}}$ (\cite{musco2017recursive}, Theorem 16) & $\Omega\parens*{\tfrac{nk}{\eps}}$ (this work, Theorem \ref{thm:main-lower-bound}) \\
    \hline
  \end{tabular}
  \caption{Table of upper bounds and lower bounds on the kernel query complexity.}
\end{figure}

Although our lower bounds show that existing upper bounds for kernel ridge regression and kernel $k$-means clustering are optimal, up to logarithmic factors, in their query complexity, one could hope that for important input distributions that may occur in practice, that better query complexities are possible. We show specifically in the case of kernel $k$-means that when the $n$ points are drawn from a mixture of $k$ Gaussians with $1/\poly(k/\eps)$ mixing probabilities and a separation between their means that matches the information-theoretically best possible for learning the means given by \cite{regev2017learning}, one can bypass the $\Omega(nk/\eps)$ lower bound, achieving an $(n/\eps) \poly(\log(k/\eps))$ query upper bound, effectively saving a factor of $k$ from the lower bounds for worst-case input distributions. This is our Theorem \ref{thm:cluster-gaussians-main}.

\subsection{Our techniques}
To prove our lower bounds, we design hard input distributions of kernel matrices as inner product matrices of an i.i.d.\ sample of vectors.

For our KRR lower bound, we draw our sample of vectors as follows: with probability $1/2$, we draw our vector uniformly from the first $(1/2)(k/\eps)$ standard basis vectors, and with probability $1/2$, we draw our vector uniformly from the next $(1/4)(k/\eps)$ standard basis vectors. Now if we draw our data set as $n$ points sampled from this distribution, then, on average, half of the input points have $n\eps/k$ copies of themselves in the data set while the other half have $2n\eps/k$ copies. We first show that correctly deciding between these two cases for a constant fraction of the $n$ input points requires $\Omega(nk/\eps)$ queries by standard arguments. We will then show that running KRR with a regularization of $\lambda = n/k$ and a target vector of all ones can solve this problem, while having an effective statistical dimension of $\Theta(k)$, giving the desired theorem. To see this, first note that the kernel matrix $\bfK$ has rank $(3/4)(k/\eps)$, where each of the $j$th nonzero eigenvalue $n_j$ is the number of copies of the $j$th standard basis vector in the data set. Then, we show that the true argmin of the KRR objective has the $i$th coordinate as $(n_j+\lambda)^{-1}$, where $n_j$ is the number of copies of the $i$th input vector. Then, if $n_j = n\eps/k$, then this is $(k/n)/(1+\eps)$ while if $n_j = 2n\eps/k$, then this is $(k/n)/(1+2\eps)$. Since these two cases are separated by a $(1\pm\eps)$ factor, a $(1+\eps)$-relative approximation of the argmin can distinguish these cases for a constant fraction of coordinates by averaging.

For our KKMC lower bound, we draw our sample of vectors as follows: we first divide the coordinates of $\mathbb R^{k/\eps}$ into $k$ blocks of size $1/\eps$, uniformly select a block, uniformly select a pair of coordinates $j_1\neq j_2$ from the block, and draw the sum of the corresponding standard basis vectors $\bfe_{j_1}+\bfe_{j_2}$. Intuitively, an optimal clustering should cluster points in the same block together, and it turns out that this clustering has a cost of $n(1-2\eps)$. We first show that for any set $S$ of size at most $\abs{S}\leq 2n/5$ points, there is a lower bound on the cost of at least $\abs{S} - (77/40)n\eps$, and that for the set $S'$ of points $\bfx$ belonging to a cluster in which uniformly sampling a point $\bfx'$ in its cluster has $\angle*{\bfx,\bfx'}\neq 0$ with probability less than $o(\eps)$, then the cost is $\abs*{S'}(1-o(\eps))$. Then setting $S$ to be the complement of $S'$, i.e.\ points with an $\Omega(\eps)$ probability of sampling a nonzero inner product, we conclude that if $\abs{S'}\geq 3n/5$, then the cost is not within a $(1+\eps/40)$ factor of the optimal cost. Thus, $\abs*{S'}\leq 3n/5$ and so for at least a $2n/5$ fraction of points, there must be an $\Omega(\eps)$ probability of sampling a nonzero inner product among its cluster. However, we then show that constructing a clustering with this guarantee requires $\Omega(nk/\eps)$ inner products by standard arguments, giving the theorem.

In our algorithm for mixtures of Gaussians, we exploit the input distribution itself to efficiently compute sketches $\bfS\bfx$ of the input points $\bfx$, where $\bfS$ is a matrix of zero mean i.i.d.\ Gaussians. Once we have computed these sketches, we show that we may compute an approximately optimal clustering in no more inner product evaluations.

\subsection{Open questions}
We suggest several related open questions. First, the error guarantee that we consider for the KRR lower bound does not directly measure the predictive performance of the KRR estimator. Thus, a more desirable result would be to find tight lower bounds for an algorithm outputting an estimator that guarantees, say, a $(1+\eps)$ relative error guarantee for the statistical risk of the resulting estimator. This is the main error guarantee considered in \cite{musco2017recursive} as well. Another interesting direction is to characterize the complexity of finding KRR estimators with objective function value guarantees as well, for which there are no query complexity efficient algorithms to the best of our knowledge.

A couple of other kernel problems have query complexity efficient algorithms using ridge leverage score sampling, including kernel principal component analysis and kernel canonical correlation analysis \cite{musco2017recursive}. We leave it open to determine whether these problems have matching lower bounds as well.

\subsection{Paper outline}
In section \ref{section:preliminaries}, we recall basic definitions and results about KRR and KKMC that we use in our lower bound results. We then prove our KRR lower bound in section \ref{section:krr-lower-bound} and our KKMC lower bound in section \ref{section:kkmc-lower-bound}. Finally, our query complexity efficient clustering algorithm for mixtures of Gaussians is given in section \ref{section:mog-upper-bound}.

\section{Preliminaries}\label{section:preliminaries}
\subsection{Notation}
We denote the set $\{1,2,\dots,n\}$ by $[n]$. For $j\in[d]$, we write $\bfe_j\in\mathbb R^d$ for the standard Euclidean basis vectors. We write $\bfI_n\in\mathbb R^{n\times n}$ for the $n\times n$ identity matrix and $\mathbf{1}_n\in\mathbb R^n$ for the vector of all ones in $n$ dimensions.

Let $\mathcal S$ be a finite set. Given two distributions $\mu,\nu$ on $\mathcal S$, the \emph{total variation distance between $\mu$ and $\nu$} is
\begin{eqn}
  D_{TV}(\mu,\nu) = \sum_{s\in\mathcal S}\abs*{\mu(s)-\nu(s)}.
\end{eqn}
We write $\Unif(\mathcal S)$ for the uniform distribution on $\mathcal S$.

Let $\mathcal X$ be the input space of a data set and $\mathcal F$ a reproducing kernel Hilbert space with kernel $k:\mathcal X\times\mathcal X\to\mathbb R$. We write $\varphi:\mathcal X\to\mathcal F$ for the feature map, i.e.\ the $\varphi$ such that $k(\bfx,\bfy) = \angle*{\varphi(\bfx),\varphi(\bfy)}_{\mathcal F}$. For a set of vectors $\{\bfx_{i}\}_{i=1}^n\subseteq\mathcal X$ and a kernel map $k:\mathcal X\times\mathcal X\to\mathbb R$, we write $\bfK\in\mathbb R^{n\times n}$ for the kernel matrix, i.e.\ the matrix with $\bfe_i^\top\bfK\bfe_j\coloneqq k(\bfx_i,\bfx_j)$. Note that $\bfK$ is symmetric and positive semidefinite (PSD). We refer the reader to \cite{scholkopf2001learning} for more details on the general theory of kernel methods. For all of our lower bound constructions, we will take $\mathcal X = \mathbb R^d$ and our kernel to be the linear kernel, i.e.\ the standard dot product on $\mathbb R^d$, $k(\bfx_i,\bfx_j) = \bfx_i\cdot\bfx_j$. Hence, we will frequently refer to kernel queries alternatively as inner product queries.

\subsection{Kernel ridge regression}
The kernel ridge regression (KRR) task is defined as follows. We parameterize an instance of KRR by a triple $(\bfK,\bfz,\lambda)$, where $\bfK\in\mathbb R^n$ is the kernel matrix of a data set $\{\bfx_i\}_{i=1}^n$, $\bfz\in\mathbb R^n$ is the target vector, and $\lambda$ is the regularization parameter. The problem is to compute
\begin{eqn}
  \bfalpha_\opt\coloneqq \argmin_{\bfalpha\in\mathbb R^n}\norm*{\bfK\bfalpha - \bfz}_2^2 + \lambda\bfalpha^\top\bfK\bfalpha.
\end{eqn}
It is well-known that the solution to the above is given in closed form by
\begin{eqn}
  \bfalpha_\opt = \parens*{\bfK+\lambda\bfI_n}^{-1}\bfz
\end{eqn}
which can be shown for example by completing the square.

An important parameter to the KRR instance $(\bfK,\bfz,\lambda)$ is the \emph{effective statistical dimension}:
\begin{dfn}[Effective statistical dimension (\cite{friedman2001elements, zhang2005learning}]
  Given a rank $r$ kernel matrix $\bfK$ with eigenvalues $\sigma_i^2$ for $i\in[r]$ and a regularization parameter $\lambda$, we define the effective statistical dimension as
  \begin{eqn}
    d_\eff^\lambda(\bfK)\coloneqq \tr\parens*{\bfK\parens*{\bfK+\lambda\bfI_n}^{-1}} = \sum_{i=1}^{r}\frac{\sigma_i^2}{\sigma_i^2+\lambda}.
  \end{eqn}
  We simply write $d_\eff^\lambda$ when the kernel matrix $\bfK$ is clear from context.
\end{dfn}
The effective statistical dimension was first introduced to measure the statistical capacity of the KRR instance, but has since been used to parameterize its computational properties as well, in the form of bounds on sketching dimension \cite{avron2017sharper} and sampling complexity \cite{alaoui2015fast, musco2017recursive}.

\subsubsection{Approximate solutions}
In the literature, various notions of approximation guarantees for KRR have been studied, including $(1+\eps)$ relative error approximations in the objective function cost \cite{avron2017sharper} and $(1+\eps)$ relative error approximations in the statistical risk \cite{bach2013sharp, alaoui2015fast, musco2017recursive}. In our paper, we consider a slightly different approximation guarantee, namely a $(1+\eps)$ relative error approximation of the argmin of the KRR objective function.
\begin{dfn}[$(1+\eps)$-approximate solution to kernel ridge regression]\label{dfn:krr-guarantee}
Given a kernel ridge regression instance $(\bfK,\bfz,\lambda)$, we say that $\hat\bfalpha\in\mathbb R^n$ is a $(1+\eps)$-approximate solution to kernel ridge regression if
\begin{eqn}\label{eqn:krr-guarantee}
    \norm*{\hat\bfalpha - \bfalpha_\opt}_2\leq \eps\norm*{\bfalpha_\opt}_2 = \eps\norm*{(\bfK+\lambda \bfI_n)^{-1}\bfz}_2.
\end{eqn}
\end{dfn}
This approximation guarantee is natural, and we note that it is achieved by the estimator of \cite{musco2017recursive}. This guarantee is then used to prove their main statistical risk guarantee. We will briefly reproduce the proof of this fact from Theorem 15 of \cite{musco2017recursive} below for completeness. Indeed, using a spectral approximation $\tilde \bfK$ satisfying $\bfK-\tilde\bfK\preceq \lambda\eps\bfI_n$, they output an estimator $\hat\bfalpha\coloneqq (\tilde \bfK+\lambda \bfI_n)^{-1}\bfz$ which satisfies the guarantee of equation (\ref{eqn:krr-guarantee}) since
\begin{eqn}
    \norm*{\hat\bfalpha-\bfalpha_\opt}_2 &= \norm*{(\tilde \bfK+\lambda \bfI_n)^{-1}\bfz - (\bfK+\lambda \bfI_n)^{-1}\bfz}_2 \\
    &= \norm*{(\tilde \bfK+\lambda \bfI_n)^{-1}[(\bfK+\lambda \bfI_n)-(\tilde \bfK+\lambda \bfI_n)](\bfK+\lambda \bfI_n)^{-1}\bfz}_2 \\
    &= \norm*{(\tilde \bfK+\lambda \bfI_n)^{-1}[\bfK-\tilde \bfK](\bfK+\lambda \bfI_n)^{-1}\bfz}_2 \\
    &\leq \norm*{(\tilde \bfK+\lambda \bfI_n)^{-1}}_2\norm*{\bfK-\tilde \bfK}_2\norm*{(\bfK+\lambda \bfI_n)^{-1}\bfz}_2 \\
    &\leq \frac1\lambda(\lambda\eps)\norm*{\bfalpha_\opt}_2 = \eps\norm*{\bfalpha_\opt}_2.
\end{eqn}

\subsection{Kernel \texorpdfstring{$k$}{k}-means clustering}
Recall the feature map $\varphi:\mathcal X\to\mathcal F$ for an input space $\mathcal X$ and a reproducing kernel Hilbert space $\mathcal F$. The problem of kernel $k$-means clustering (KKMC) involves forming a partition of the data set $\{\bfx_i\}_{i=1}^n$ into $k$ clusters $\mathcal C\coloneqq \{C_j\}_{j=1}^k$ with centroids $\bfmu_j\coloneqq \tfrac1{\abs{C_j}}\sum_{\bfx\in C_j}\varphi(\bfx)$ such that the objective function
\begin{eqn}
  \cost(\mathcal C)\coloneqq \sum_{j=1}^k \sum_{\bfx\in C_j}\norm*{\varphi(\bfx)-\bfmu_j}_{\mathcal F}^2
\end{eqn}
is minimized. The problem of finding exact solutions are known to be NP-hard \cite{aloise2009np}, but it has nonetheless proven to be an extremely popular model in practice \cite{hartigan1975clustering}.

With an abuse of notation, we will also talk about the cost of a single cluster, which is just the above sum taken only over one cluster:
\begin{eqn}
  \cost(C_j)\coloneqq \sum_{\bfx\in C_j}\norm*{\varphi(\bfx)-\bfmu_j}_{\mathcal F}^2.
\end{eqn}

As done in \cite{boutsidis2009unsupervised, cohen2015dimensionality, musco2017recursive} and many other works, we consider the approximation guarantee of finding a clustering that achieves a $(1+\eps)$ relative error in the objective function cost, i.e.\ a finding a partition $\{C_j'\}_{j=1}^k$ such that
\begin{eqn}
  \cost(\{C_j'\}_{j=1}^k) &= \sum_{j=1}^k \sum_{\bfx\in C_j'}\norm*{\varphi(\bfx)-\bfmu_j}_{\mathcal F}^2 \\
  &\leq (1+\eps)\min_{\mathcal C}\cost(\mathcal C) = (1+\eps)\min_{\mathcal C=\{C_j\}_{j=1}^k}\sum_{j=1}^k \sum_{\bfx\in C_j}\norm*{\varphi(\bfx)-\bfmu_j}_{\mathcal F}^2.
\end{eqn}

\section{Lower bound for kernel ridge regression}\label{section:krr-lower-bound}
We present our lower bound on the number of kernel entries required in order to compute a $(1+\eps)$-approximate solution to kernel ridge regression (see definition \ref{dfn:krr-guarantee}).

\begin{thm}[Query lower bound for kernel ridge regression]\label{thm:krr-hardness-main}
Consider a possibly randomized algorithm $\mathcal A$ that correctly outputs a $(1+\eps)$-approximate solution $\hat\bfalpha\in\mathbb R^n$ (see definition \ref{dfn:krr-guarantee}) to any kernel ridge regression instance $(\bfK,\bfz,\lambda)$ with probability at least $2/3$. Then there exists an input instance $(\bfK,\bfz,\lambda)$ on which $\mathcal A$ reads at least $\Omega(nd_\eff^\lambda/\eps)$ entries of $\bfK$, possibly adaptively, in expectation.
\end{thm}

Our lower bound is nearly optimal, matching the ridge leverage score algorithm in Theorem 15 of \cite{musco2017recursive} which reads $O\parens*{\tfrac{nd_\eff^\lambda}{\eps}\log\tfrac{nd_\eff^\lambda}{\eps}}$ kernel entries up to a $\log\tfrac{d_\eff^\lambda}{\eps}$ factor.

\subsection{Main lower bound}
\begin{dfn}[Hard input distribution -- kernel ridge regression]\label{dfn:hard-krr}
Let $J,n\in\mathbb N$ and assume for simplicity that $4\mid J$. We define a distribution $\mu_\KRR(n,J)$ on binary PSD matrices $\bfK\in\mathbb R^{n\times n}$ defined as follows. We first define a distribution $\nu_\KRR(J)$ over standard basis vectors $\{\bfe_j\in\mathbb R^{3J/4} : j\in[3J/4]\}$, where with probability $1/2$ we draw a uniformly random $\bfe_j$ from $S_1\coloneqq\{\bfe_j:j\in[J/2]\}$ and with probability $1/2$ we draw a uniformly random $\bfe_j$ from $S_2\coloneqq\{\bfe_{j+J/2} : j\in[J/4]\}$. We then generate $\bfK$ by drawing $n$ i.i.d.\ samples $\{\bfx_i\}_{i=1}^n$ from $\nu_\KRR(J)$ and letting $\bfK$ be the inner product matrix of $\{\bfx_i\}_{i=1}^n$, that is, $\bfe_i^\top\bfK\bfe_j\coloneqq \bfx_i\cdot \bfx_j$.
\end{dfn}

\begin{thm}\label{thm:krr-hardness}
Let $\eps\in(0,1/2)$ and $J = k/\eps$ with $J^2 = O(n)$ and $k$ a parameter. Suppose that there exists a possibly randomized algorithm $\mathcal A$ that, with probability at least $2/3$ over its random coin tosses and random kernel matrix drawn from $\bfK\sim\mu_\KRR(n,J)$, correctly outputs a $(1+\eps/100)$-approximate solution $\hat\bfalpha\in\mathbb R^n$ (see definition \ref{dfn:krr-guarantee}) to the kernel ridge regression instance $(\bfK,\bfz,\lambda)$ with $\bfz = \mathbf{1}_n$ and $\lambda = n/k$. Furthermore, suppose that $\mathcal A$ reads at most $r$ positions of $\bfK$ on any input, possibly adaptively. Then, $d_\eff^\lambda(\bfK) = \Theta(k)$ and $r = \Omega(nd_\eff^\lambda/\eps)$.
\end{thm}

To prove Theorem \ref{thm:krr-hardness}, we will make a reduction to the following hardness lemma.

\begin{lem}\label{lem:krr-hardness}
Recall the definitions of $\mu_\KRR(n,J), \nu_\KRR(J), S_1, S_2$ from definition \ref{dfn:hard-krr}. Suppose that there exists a possibly randomized algorithm $\mathcal A$ that, with probability at least $2/3$ over its random coin tosses and random inputs drawn from $\mu_\KRR(n,k/\eps)$, correctly outputs whether $\bfx_{i}$ corresponds to $\bfe_j$ with $j\in S_1$ or $j\in S_2$ for at least a $9/10$ fraction of rows $\bfe_i^\top\bfK$ for $i\in[n]$. Further, suppose that $\mathcal A$ reads at most $r$ positions of $\bfK$ on any input, possibly adaptively. Then, $r = \Omega(nJ)$.
\end{lem}

\begin{figure}[ht]
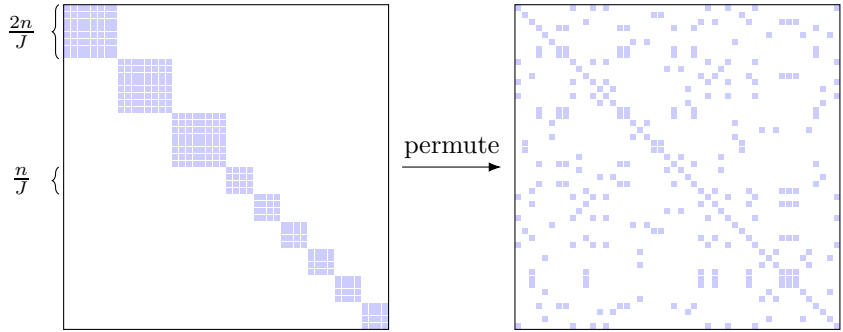

    \centering
    \begin{tikzpicture}[scale=0.3]
        \draw [decorate,decoration={brace},xshift=-6pt,yshift=0pt] (0,12) -- (0,14.4) node [black,midway,xshift=-0.5cm]  {$\frac{2n}{J}$};

        \begin{scope}
        \input{unpermuted.tex}
        \end{scope}

        \draw [decorate,decoration={brace},xshift=-6pt,yshift=0pt] (0,6) -- (0,7.2) node [black,midway,xshift=-0.5cm]  {$\frac{n}{J}$};

        \draw[-{Latex[]}] (15,7.2) -- (19.4,7.2) node[midway,anchor=south] {permute};

        \begin{scope}[shift={(20,0)}]
        \input{permuted.tex}
        \end{scope}
    \end{tikzpicture}
    \caption{The hardness lemma -- does the $i$th row have $\tfrac{2n}{J}$ or $\tfrac{n}{J}$ ones?}
    \label{fig:krr-hardness-lemma}
\end{figure}

\begin{proof}[Proof of Lemma \ref{lem:krr-hardness}]
First consider a single draw $\bfx\sim\nu_\KRR(J)$. We claim that $\Omega(J)$ adaptive inner product queries $\bfx\cdot\bfe_j$ are required to correctly output whether $\bfx\in S_1$ or $\bfx\in S_2$ with probability at least $2/3$ over $\nu_\KRR(J)$. Suppose there exists a randomized algorithm $\mathcal B$ that has the above guarantee. By Yao's minimax principle \cite{yao1977probabilistic}, there then exists a deterministic algorithm $\mathcal B'$ with the same guarantee and the same expected cost over the input distribution. Then, $\mathcal B'$ can be used to construct a hypothesis test to decide whether $\bfx\sim \Unif(S_1)$ or $\bfx\sim \Unif(S_2)$ which succeeds with probability at least $2/3$. Now let $S$ denote the random variable indicating the list of inner product queries made and their corresponding values, let $L_1$ denote the distribution of $S$ conditioned on $\bfx\sim \Unif(S_1)$, and $L_2$ the distribution of $S$ conditioned on $\bfx\sim \Unif(S_2)$. Then by Proposition 2.58 of \cite{bar2002complexity}, we have that
\begin{eqn}
    \frac{1+D_{TV}(L_1,L_2)}{2}\geq\frac23
\end{eqn}
and thus rearranging gives $D_{TV}(L_1,L_2)\geq 1/3$. Now suppose for contradiction that $\mathcal B'$ makes at most $q\leq J/100$ queries on any input. Since $\mathcal B'$ is deterministic, it makes the same sequence of inner product queries $\bfx\cdot\bfe_{j_1},\bfx\cdot\bfe_{j_2},\dots,\bfx\cdot\bfe_{j_q},$ if it reads a sequence of $q$ zeros. Now fix these queries $j_1,j_2,\dots,j_q$. We then have that for each $\ell\in[q]$,
\begin{eqn}
    \Pr_{\bfx\sim\Unif(S_1)}\parens*{\bfx = \bfe_{j_\ell}} = \frac1J,\qquad\Pr_{\bfx\sim\Unif(S_2)}\parens*{\bfx = \bfe_{j_\ell}} = \frac2J
\end{eqn}
and thus by the union bound,
\begin{eqn}
    \Pr_{\bfx\sim\Unif(S_1)}\parens*{\bfx\in\braces*{\bfe_{j_\ell} : \ell\in[q]}} \leq \frac{q}J,\qquad\Pr_{\bfx\sim\Unif(S_2)}\parens*{\bfx\in\braces*{\bfe_{j_\ell} : \ell\in[q]}} \leq \frac{2q}J.
\end{eqn}
Now let $\Omega$ denote the support of $S$ and let $s_0\in\Omega$ denote the value of $S$ when $\mathcal B'$ reads all zeros. Then,
\begin{eqn}
    D_{TV}(L_1,L_2) &= \sum_{s\in\Omega}\abs*{\Pr_{\bfx\sim\Unif(S_1)}\parens*{S=s} - \Pr_{\bfx\sim\Unif(S_2)}\parens*{S=s}} \\
    &= \abs*{\Pr_{\bfx\sim\Unif(S_1)}\parens*{S=s_0} - \Pr_{\bfx\sim\Unif(S_2)}\parens*{S=s_0}} + \sum_{s\in\Omega\setminus\{s_0\}}\abs*{\Pr_{\bfx\sim\Unif(S_1)}\parens*{S=s} - \Pr_{\bfx\sim\Unif(S_2)}\parens*{S=s}} \\
    &\leq \frac{2q}{J} + \sum_{s\in\Omega\setminus\{s_0\}}\Pr_{\bfx\sim\Unif(S_1)}\parens*{S=s} + \Pr_{\bfx\sim\Unif(S_2)}\parens*{S=s} \\
    &= \frac{2q}{J} + \Pr_{\bfx\sim\Unif(S_1)}\parens*{S\neq s_0} + \Pr_{\bfx\sim\Unif(S_2)}\parens*{S\neq s_0}\leq \frac{2q}{J} + \frac{q}{J} + \frac{2q}{J} = \frac{5q}{J}\leq \frac1{20}
\end{eqn}
which contradicts $D_{TV}(L_1,L_2)\geq 1/3$. Thus, we conclude that $q > J/100$.

We now prove the full claim via a reduction to the above problem of deciding whether some $\bfx\sim\nu_\KRR(J)$ is either drawn from $S_1$ or $S_2$. Suppose for contradiction that there exists a randomized algorithm $\mathcal A$ with the guarantees of the lemma which reads $r = o(nJ)$. We then design an algorithm $\mathcal B$ using $\mathcal A$ as follows. We independently sample a uniformly random index $i^*\sim\Unif([n])$ and $n-1$ points $\{\bfx_i\}_{i=1}^{n-1}$ with $\bfx_i\sim\nu_\KRR(J)$ for each $i\in[n-1]$. We then run $\mathcal A$ on the kernel matrix instance $\bfK$ corresponding to setting the $i^*$th standard basis vector to $\bfx$ and the other $n-1$ vectors according to $\{\bfx_i\}_{i=1}^{n-1}$. Note then that we can generate any entry of $\bfK$ on row $i^*$ or column $i^*$ by an inner product query to $\bfx$, and otherwise we can simulate the kernel query without making any inner product queries to $\bfx$. If $\mathcal A$ ever reads more than $J/100$ entries of $\bfx$, we output failure. Since $r = o(nJ)$, by averaging, for at least a $199/200$ fraction of the $n$ rows of $\bfK$, $\mathcal A$ reads at most $J/200$ entries of the row $\bfe_i^\top\bfK$. Similarly, for at least a $199/200$ fraction of the $n$ columns of $\bfK$, $\mathcal A$ reads at most $J/200$ entries of the column $\bfK\bfe_i$. Thus, for at least a $99/100$ fraction of the input points, $\mathcal A$ makes at most $J/100$ inner product queries. It follows by symmetry that with probability $99/100$, $\mathcal A$ makes at most $J/100$ inner product queries on $\bfx$. Then by a union bound over the random choice of $i^*$ over the $n$ input points, $\mathcal A$ correctly decides whether $\bfx\sim\Unif(S_1)$ or $\bfx\sim\Unif(S_2)$ and attempts to read at most $J/100$ entries of $\bfx$ with probability at least $1/10 + 1/100 = 11/100$. Thus, $\mathcal B$ succeeds with probability at least $1-11/100\geq 2/3$, contradicting the above result.
\end{proof}

With Lemma \ref{lem:krr-hardness} in hand, we finally get to the proof of Theorem \ref{thm:krr-hardness}.

\begin{proof}[Proof of Theorem \ref{thm:krr-hardness}]
Assume that $nJ = o(n^2)$, since otherwise the lower bound is $\Omega(n^2)$, which is best possible. Note that for $\bfx\sim\nu_\KRR(J)$, $\bfx = \bfe_j$ with probability $\tfrac12\tfrac1{J/2} = \tfrac1J$ if $\bfe_j\in S_1$ and $\tfrac12\tfrac1{J/4} = \tfrac2J$ if $\bfe_j\in S_2$. For a fixed $j\in[3J/4]$, let $n_j$ be the number of $\bfe_j$ sampled in $\bfK$ and $\mu_j\coloneqq \E_{\bfK\sim\mu_\KRR(n,J)}(n_j)$. Note that $\mu_j = n/J$ for $j\in[J/2]$ and $\mu_j = 2n/J$ for $j\in[J/4]+J/2$. Then by Chernoff bounds,
\begin{eqn}
    \Pr_{\bfK\sim\mu_\KRR(n,J)}\parens*{\braces*{\abs*{n_j - \mu_j}\geq \frac1{100}\mu_j}}\leq 2\exp\parens*{-\frac1{100}\frac{\mu_j}{3}}\leq 2\exp\parens*{-\frac1{100}\frac{n/J}{3}}
\end{eqn}
so by a union bound, we have that
\begin{eqn}
    \Pr_{\bfK\sim\mu_\KRR(n,J)}\parens*{\bigcup_{j\in[3J/4]}\braces*{\abs*{n_j - \mu_j}\geq \frac1{100}\mu_j}}\leq 2\frac{3J}{4}\exp\parens*{-\frac1{100}\frac{n/J}{3}}.
\end{eqn}
Since $nJ = o(n^2)$, we have that $n/J = \omega(1)$. Furthermore, since $J^2 = O(n)$, we have that $J = O(n/J)$. Thus, the above happens with probability at most $1/100$ by taking $n/J$ large enough. Dismiss this event as a failure and assume that $\abs{n_j-\mu_j}\leq\tfrac1{100}\mu_j$ for all $j\in[3J/4]$.

Now let $\bfK = \bfU\bfSigma\bfU^\top$ be the full SVD of $\bfK$. Note that the first $3J/4$ singular values are $n_j$ with corresponding singular vectors $\bfU\bfe_j = \tfrac1{\sqrt{n_j}}\bfK\bfe_j$, and the rest are all $0$s. Then, the target vector $\bfz = \mathbf{1}_n$ can be written as
\begin{eqn}
    \bfz = \sum_{j\in[3J/4]}\bfK\bfe_j = \sum_{j\in[3J/4]}\sqrt{n_j}\bfU\bfe_j,
\end{eqn}
since each coordinate $i\in[n]$ belongs to exactly one of the $3J/4$ input points drawn from $\nu_\KRR(n,J)$. The optimal solution can then be written as
\begin{eqn}
    \bfalpha_\opt &= (\bfK+\lambda\bfI_n)^{-1}\bfz = \bfU(\bfSigma+\lambda\bfI_n)^{-1}\bfU^\top\bfz \\
    &= \sum_{j\in[3J/4]}\sqrt{n_j}\bfU(\bfSigma+\lambda\bfI_n)^{-1}\bfU^\top\bfU\bfe_j = \sum_{j\in[3J/4]}\frac{1}{n_j+\lambda}\parens*{\sqrt{n_j}\bfU\bfe_j}.
\end{eqn}
Thus, for $i\in[n]$, the optimal solution takes the value $(\bfalpha_\opt)_i = (n_{j_i}+\lambda)^{-1}$ where $j_i\in[3J/4]$ is the index of the standard basis vector that the $i$th input point corresponds to.

Now by multiplying the $(1+\eps/100)$-approximation guarantee by $n/k$ and squaring, we have that
\begin{eqn}
    \norm*{\frac{n}{k}\hat\bfalpha - \frac{n}{k}\bfalpha_\opt}_2^2 &\leq \frac{\eps^2}{100^2}\norm*{\frac{n}{k}\bfalpha_\opt}_2^2 = \frac{\eps^2}{100^2}\sum_{j\in[3J/4]}\norm*{\frac{n/k}{n_j+\lambda}\parens*{\sqrt{n_j}\bfU\bfe_j}}_2^2\leq \frac{\eps^2}{100^2}\norm*{\bfz}_2^2 = \frac{\eps^2}{100^2}n
\end{eqn}
so by averaging, we have that $\parens*{\tfrac{n}{k}(\hat\bfalpha)_i-\tfrac{n}{k}(\bfalpha_\opt)_i}^2\leq \eps^2/100$ for at least a $99/100$ fraction of the $n$ coordinates of $i$. Then on these coordinates, $\abs*{\tfrac{n}{k}(\hat\bfalpha)_i - \tfrac{n}{k}(\bfalpha_\opt)_i}\leq \eps/10$. Now note that on these coordinates, we have that
\begin{eqn}
    \abs*{\frac{n}{k}(\hat\bfalpha)_i - \frac{n}{k}\frac1{\mu_j+\lambda}} &\leq \abs*{\frac{n}{k}(\hat\bfalpha)_i - \frac{n}{k}(\bfalpha_\opt)_i} + \abs*{\frac{n}{k}(\bfalpha_\opt)_i - \frac{n}{k}\frac1{\mu_j+\lambda}} \\
    &\leq \frac{\eps}{10} + \frac{n}{k}\abs*{\frac1{n_j+n/k}-\frac1{\mu_j+n/k}} \leq \frac{\eps}{10} + \frac{n}{k}\frac{\abs{n_j-\mu_j}}{(n_j+n/k)(\mu_j+n/k)} \\
    &\leq \frac{\eps}{10} + \frac{\mu_j/100}{n/k}\leq \frac{\eps}{10} + \frac{2n\eps/(100k)}{n/k} = \frac{6}{50}\eps.
\end{eqn}
Since
\begin{eqn}
    \frac{n}{k}\frac1{n\eps/k + n/k} - \frac{n}{k}\frac1{2n\eps/k + n/k} = \frac1{1+\eps} - \frac1{1+2\eps} = \frac\eps{(1+\eps)(1+2\eps)} > \frac{\eps}3 > 2\frac{6}{50}\eps
\end{eqn}
for $\eps\in(0,1/2)$, we can distinguish whether the $i$th input point has $\mu_j = n\eps/k$ or $\mu_j = 2n\eps/k$ on these coordinates and thus we can solve the hard computational problem of Lemma \ref{lem:krr-hardness} without reading anymore entries of $\bfK$ after solving the kernel ridge regression instance. Thus, we have that $\mathcal A$ reads $\Omega(nk/\eps)$ kernel entries by a reduction to Lemma \ref{lem:krr-hardness}.

Finally, to obtain the statement of the theorem, it remains to show that $d_\eff^\lambda = \Theta(k)$. Indeed,
\begin{eqn}
    d_\eff^\lambda = \sum_{j\in[3J/4]}\frac{n_j}{n_j+\lambda} = \Theta\parens*{\sum_{j\in[3J/4]}\frac{n\eps/k}{n\eps/k+n/k}} = \Theta\parens*{k}
\end{eqn}
as desired.
\end{proof}

We now obtain Theorem \ref{thm:krr-hardness-main} by scaling parameters by constant factors.

\begin{rem}
The setting of the regularization parameter in the above construction is a bit unnatural as the top $d_\eff^\lambda = \Theta(k)$ singular values of the kernel matrix are of order $n\eps/k$ while the regularization is of order $n/k$, which is $1/\eps$ times larger. One can easily fix this as follows. We add $(n/k)\bfe_i$ to the end of our data set for $i = k/\eps + 1, k/\eps + 2,\dots, k/\eps + k$. This only increases our effective statistical dimension to
\begin{eqn}
    d_\eff^\lambda = \sum_{j\in[3J/4]}\frac{n_j}{n_j+\lambda} + \sum_{i=1}^k \frac{n/k}{n/k + \lambda} = \Theta\parens*{\sum_{j\in[3J/4]}\frac{n\eps/k}{n\eps/k+n/k} + \frac{k}{2}} = \Theta\parens*{k}
\end{eqn}
and our hardness argument is clearly unaffected. Now the setting of the regularization is such that it scales as the top $d_\eff^\lambda$ singular values, so that it reduces the effects of the next $k/\eps$ noisy directions, which is natural.
\end{rem}

\subsection{Extensions to other kernels}
The above lower bound was proven just for the dot product kernel, but we note that essentially the same proof applies to more general kernels as well. To this end, we introduce the notion of \emph{indicator kernels}:
\begin{dfn}[Indicator kernels]\label{dfn:indicator-kernels}
We say that $k:\mathbb R^d\times\mathbb R^d\to\mathbb R$ is an \emph{indicator kernel} if there exist $c_1>0$ and $c_0<c_1$ such that
\begin{eqn}
    k(\bfe_i,\bfe_j) =
    \begin{cases}
        c_1 & \text{if $i=j$} \\
        c_0 & \text{otherwise}
    \end{cases}
\end{eqn}
for all standard basis vectors $\bfe_i,\bfe_j$ for $i,j\in[d]$.
\end{dfn}
Examples of such kernels include generalized dot product kernels and distance kernels, i.e.\ kernels of the form $k(\bfx,\bfx') = f(\bfx\cdot\bfx')$ and $k(\bfx,\bfx') = f(\norm*{\bfx-\bfx'}_2)$ for an appropriate function $f:\mathbb R\to\mathbb R$, which in turn include important kernels such as the polynomial kernel, the Gaussian kernel, etc.

Note that if $c_0 = 0$, the kernel matrix is just $c_1$ times the kernel matrix from before, so it is easy to see that the exact same proof works after scaling $\lambda$ by $c_1$. When $c_0$ is nonzero, then every entry of the kernel matrix is offset by $c_0$. However we will see that even in this case, the same proof still applies.
\begin{thm}[Query lower bound for kernel ridge regression for indicator kernels]\label{thm:krr-hardness-main-indicator}
The lower bound of Theorem \ref{thm:krr-hardness-main} continues to hold for any algorithm computing a $(1+\eps)$ relative error solution to a KRR instance with an indicator kernel (Definition \ref{dfn:indicator-kernels}) instead of the dot product kernel.
\end{thm}

\begin{proof}
Suppose we draw our kernel $\bfK$ as in Definition \ref{dfn:hard-krr}, with the dot product kernel being replaced by any indicator kernel. Let $\bfG$ be the inner product matrix of the point set with SVD $\bfG = \bfU\bfSigma\bfU^\top$ as before. Then, we may write the kernel matrix as
\begin{eqn}
    \bfK = c_0\mathbf{1}_{n\times n} + (c_1-c_0)\bfG.
\end{eqn}
Now define
\begin{eqn}
    C\coloneqq c_0\mathbf{1}_n((c_1-c_0)\bfG+\lambda\bfI_n)^{-1}\bf1_n.
\end{eqn}
Then, by the \href{https://en.wikipedia.org/wiki/Sherman\%E2\%80\%93Morrison\_formula}{Sherman-Morrison formula}, $C\neq -1$ since $(\bfK+\lambda\bfI_n)$ is invertible, and so we have that
\begin{eqn}
    \bfalpha_\opt &= (\bfK+\lambda\bfI_n)^{-1}\bfz \\
    &= (c_0\mathbf{1}_{n}\mathbf{1}_n^\top + (c_1-c_0)\bfG+\lambda\bfI_n)^{-1}\bfz \\
    &= ((c_1-c_0)\bfG+\lambda\bfI_n)^{-1}\bfz - \frac{((c_1-c_0)\bfG+\lambda\bfI_n)^{-1}(c_0\mathbf{1}_{n}\mathbf{1}_n^\top)((c_1-c_0)\bfG+\lambda\bfI_n)^{-1}}{1+c_0\mathbf{1}_{n}^\top((c_1-c_0)\bfG+\lambda\bfI_n)^{-1}\mathbf{1}_{n}}\mathbf{1}_n \\
    &= ((c_1-c_0)\bfG+\lambda\bfI_n)^{-1}\bfz - ((c_1-c_0)\bfG+\lambda\bfI_n)^{-1}\mathbf{1}_{n}\frac{c_0\mathbf{1}_n^\top((c_1-c_0)\bfG+\lambda\bfI_n)^{-1}\mathbf{1}_n}{1+c_0\mathbf{1}_{n}^\top((c_1-c_0)\bfG+\lambda\bfI_n)^{-1}\mathbf{1}_{n}} \\
    &= \parens*{1 - \frac{C}{1+C}}((c_1-c_0)\bfG+\lambda\bfI_n)^{-1}\bfz \\
    &= \frac1{(c_1-c_0)(1+C)}(\bfG+(\lambda/(c_1-c_0))\bfI_n)^{-1}\bfz
\end{eqn}
Thus, we find that the exact same proof as before works by setting $\lambda = (c_1-c_0)n/k$.
\end{proof}

\section{Lower bound for kernel \texorpdfstring{$k$}{k}-means clustering}\label{section:kkmc-lower-bound}

\subsection{Finding the cost vs.\ assigning points}
Recall that \cite{musco2017recursive} present an algorithm for solving KKMC with a kernel querying complexity of $O\parens*{\tfrac{nk}{\eps}\log\tfrac{k}{\eps}}$. We now briefly present some intuition on how we would like to match this up to $\log\tfrac{k}{\eps}$. We first note that the hardness cannot come from finding the centers of an approximately optimal clustering or approximating the cost of the optimal clustering up to $(1\pm\eps)$, since there is an algorithm for finding these in $O(nk + \poly(k,1/\eps,\log n))$ kernel queries: indeed, Theorem 15.5 of \cite{feldman2011unified} shows how to find a strong $\eps$-coreset of size $\poly(k\log n/\eps)$ in $O(nk + \poly(k,1/\eps,\log n))$ kernel queries, which can then be used to compute both approximate centers and the cost. Thus, intuitively, in order to achieve a lower bound of $\Omega(nk/\eps)$ which nearly matches the dominant term in the upper bound of \cite{musco2017recursive}, we must design a hard point set in which the hardness is not in computing the cost nor the centers, but rather in assigning the $n$ input points to their appropriate clusters.

We take this opportunity to prove a lower bound of $\Omega(nk)$ kernel queries for the problem of computing a $(1+\eps)$ relative error approximation to the cost of KKMC. In practical applications, the $nk$ term dominates the $\poly(k,1/\eps,\log n)$ term and thus we obtain a fairly tight characterization of this subproblem of KKMC. To prove this result, we make use of the hardness of deciding whether a binary PSD matrix has rank $k$ or $k+1$.

\begin{dfn}[Hard input distribution -- rank]
Consider the distribution $\mu_{\rank}(n,k)$ on binary PSD matrices $\bfK\in\mathbb R^{n\times n}$ defined as follows. We first draw $n$ i.i.d.\ samples $\{\bfx_i\}_{i=1}^n$ drawn from $\Unif(\braces*{\bfe_j : j\in[k]})$. Then, with probability $1/2$, select a uniformly random index $i^*\in[n]$ and set $\bfx_i \coloneqq \bfe_{k+1}$. Finally, generate $\bfK$ as the inner product matrix of $\{\bfx_i\}_{i=1}^n$, that is, $\bfe_i^\top\bfK\bfe_j\coloneqq \bfx_i\cdot\bfx_j$.
\end{dfn}

\begin{lem}\label{lem:rank-hardness}
Suppose there exists an algorithm which, with probability at least $5/8$, over its random coin tosses and random inner product matrix $\bfK\sim\mu_{\rank}(n,k)$, correctly decides whether $\bfK$ has rank $k+1$ or at most $k$. Furthermore, suppose that the algorithm reads at most $r$ positions of $\bfK$, possibly adaptively. Then, $r = \Omega(nk)$.
\end{lem}
\begin{proof}
Suppose for contradiction that $r = o(nk)$. Let $\mathcal S_k\coloneqq\{\bfe_j:j\in[k]\}$. We consider the following hypothesis test: decide whether some $\bfx$ is drawn from $\bfx\sim\Unif(\mathcal S_k)$ or from $\bfx = \bfe_{k+1}$ using inner product queries of the form $\bfx\cdot\bfe_\ell$ for $\ell\in[k]$. By Yao's minimax principle \cite{yao1977probabilistic}, there exists a deterministic algorithm $\mathcal A'$ with the same guarantee as $\mathcal A$ and the same expected cost over the input distribution. We then design an algorithm $\mathcal B$ using $\mathcal A'$ as follows. We independently sample a uniformly random index $i^*\sim\Unif([n])$ and $n-1$ points $\{\bfx_i\}_{i=1}^{n-1}$ with $\bfx_i\sim\Unif(\mathcal S_k)$ for each $i\in[n-1]$. We then run $\mathcal A'$ on the kernel matrix instance $\bfK$ corresponding to setting the $i^*$th standard basis vector to $\bfx$ and the other $n-1$ vectors according to $\{\bfx_i\}_{i=1}^{n-1}$. Note then that we can generate any entry of $\bfK$ on row $i^*$ or column $i^*$ by an inner product query to $\bfx$, and otherwise we can simulate the kernel query without making any inner product queries to $\bfx$. If $\mathcal A'$ ever reads more than $k/100$ entries of $\bfx$, we output failure.

Note that with probability at least $99/100$ over $\{\bfx_i\}_{i=1}^{n-1}$, each $\bfe_j$ for $j\in[k]$ is drawn at least once for $n$ large enough. Thus, in this event, $\bfK$ has rank $k$ only if $\bfx\sim\Unif(\mathcal S_k)$ and $k+1$ otherwise. Since $\mathcal A'$ is correct with probability $5/8$, by a union bound, $\mathcal B$ is correct with probability at least $1-(\tfrac1{100}+\tfrac38)\geq 5/9$.

Let $S$ denote the random variable indicating the list of positions of $\bfK$ read by $\mathcal A'$ and its corresponding values,let $L_1$ denote the distribution of $S$ conditioned on $\bfx\sim\Unif(\mathcal S_k)$, and $L_2$ the distribution of $S$ conditioned on $\bfx = \bfe_{k+1}$. Then by Proposition 2.58 of \cite{bar2002complexity}, we have that
\begin{eqn}
    \frac{1+D_{TV}(L_1,L_2)}{2}\geq\frac{5}{9}
\end{eqn}
so $D_{TV}(L_1,L_2)\geq 1/9$.

Since $r = o(nk)$, by averaging, for at least a $199/200$ fraction of the $n$ rows of $\bfK$, $\mathcal A'$ reads at most $k/200$ entries of the row $\bfe_i^\top\bfK$. Similarly, for at least a $199/200$ fraction of the $n$ columns of $\bfK$, $\mathcal A'$ reads at most $k/200$ entries of the column $\bfK\bfe_i$. Thus, for at least a $99/100$ fraction of the input points, $\mathcal A'$ makes at most $k/100$ inner product queries.

If we condition on the event that $\bfx\sim\Unif(\mathcal S_k)$, we have that $\mathcal A'$ makes at most $k/100$ inner product queries with $\bfx$ with probability at least $99/100$ over the randomness of $i^*$ by symmetry. That is, if $\mathcal E$ is the event that $\mathcal A'$ makes at most $k/100$ inner product queries on row and column $i^*$, then
\begin{eqn}
    \Pr_{\substack{i^*\sim\Unif([n]) \\
    \{\bfx_i\}_{i=1}^{n-1}\sim\Unif(\mathcal S_k)^{n-1} \\
    \bfx\sim\Unif(\mathcal S_k)}}\parens*{\mathcal E}\geq\frac{99}{100}.
\end{eqn}
Now letting $\mathcal E'$ be the event that $\mathcal A'$ sees $0$s on all of its queries on row and column $i^*$, we have that
\begin{eqn}
    \Pr_{\substack{i^*\sim\Unif([n]) \\
    \{\bfx_i\}_{i=1}^{n-1}\sim\Unif(\mathcal S_k)^{n-1} \\
    \bfx\sim\Unif(\mathcal S_k)}}\parens*{\mathcal E'\mid\mathcal E} &= \sum_{w\in\Omega}\Pr_{\substack{i^*\sim\Unif([n]) \\
    \bfx\sim\Unif(\mathcal S_k)}}\parens*{\mathcal E'\mid\mathcal E, \{\bfx_i\}_{i=1}^{n-1}=w}\Pr_{\substack{i^*\sim\Unif([n]) \\
    \bfx\sim\Unif(\mathcal S_k)}}\parens*{\mathcal E'\mid\mathcal E, \{\bfx_i\}_{i=1}^{n-1}=w}.
\end{eqn}
Once we fix $\{\bfx_i\}_{i=1}^{n-1}$, there is a $1-1/k$ probability over $\bfx$ that any fixed query $\bfx\cdot\bfe_\ell$ returns a $0$, so the probability that $\mathcal E'$ happens is at least
\begin{eqn}
    \parens*{1-\frac1k}^{k/100}\geq 1-\frac1k\frac{k}{100} = \frac{99}{100}.
\end{eqn}
Thus, by the chain rule,
\begin{eqn}
    \Pr_{\substack{i^*\sim\Unif([n]) \\
    \{\bfx_i\}_{i=1}^{n-1}\sim\Unif(\mathcal S_k)^{n-1} \\
    \bfx\sim\Unif(\mathcal S_k)}}\parens*{\mathcal E'\cap\mathcal E} &= \Pr_{\substack{i^*\sim\Unif([n]) \\
    \{\bfx_i\}_{i=1}^{n-1}\sim\Unif(\mathcal S_k)^{n-1} \\
    \bfx\sim\Unif(\mathcal S_k)}}\parens*{\mathcal E'\mid\mathcal E}\Pr_{\substack{i^*\sim\Unif([n]) \\
    \{\bfx_i\}_{i=1}^{n-1}\sim\Unif(\mathcal S_k)^{n-1} \\
    \bfx\sim\Unif(\mathcal S_k)}}\parens*{\mathcal E} \\
    &\geq \frac{99}{100}\frac{99}{100}\geq \frac{49}{50}.
\end{eqn}
Bounding the event $\mathcal E'\cup\mathcal E$ by $\mathcal E'$ as sets, we have that
\begin{eqn}
    \Pr_{\substack{i^*\sim\Unif([n]) \\
    \{\bfx_i\}_{i=1}^{n-1}\sim\Unif(\mathcal S_k)^{n-1} \\
    \bfx\sim\Unif(\mathcal S_k)}}\parens*{\mathcal E'}\geq\frac{49}{50}.
\end{eqn}

 Also let $W\subseteq \supp(S)$ denote the set of $s\in \supp(S)$ such that $\mathcal A'$ reads all zeros on row and column $i^*$. Then,
\begin{eqn}
    D_{TV}(L_1,L_2) &= \sum_{s\in \supp(S)}\abs*{
        \Pr_{\substack{i^*,\{\bfx_i\}_{i=1}^{n-1} \\ \bfx\sim\Unif(\mathcal S_k)}}\parens*{S=s} -
        \Pr_{\substack{i^*,\{\bfx_i\}_{i=1}^{n-1} \\ \bfx=\bfe_{k+1}}}\parens*{S=s}
    } \\
    &= \sum_{s\in \supp(S)}\abs*{
        \parens*{
            \Pr_{\substack{i^*,\{\bfx_i\}_{i=1}^{n-1} \\ \bfx\sim\Unif(\mathcal S_k)}}\parens*{S=s\mid\mathcal E'}\Pr_{\substack{i^*,\{\bfx_i\}_{i=1}^{n-1} \\ \bfx\sim\Unif(\mathcal S_k)}}\parens*{\mathcal E'} -
            \Pr_{\substack{i^*,\{\bfx_i\}_{i=1}^{n-1} \\ \bfx=\bfe_{k+1}}}\parens*{S=s}
        }
    } \\
    &= \sum_{s\in \supp(S)}
        \Pr_{\substack{i^*,\{\bfx_i\}_{i=1}^{n-1} \\ \bfx=\bfe_{k+1}}}\parens*{S=s}
        \abs*{
            \Pr_{\substack{i^*,\{\bfx_i\}_{i=1}^{n-1} \\ \bfx\sim\Unif(\mathcal S_k)}}\parens*{\mathcal E'} - 1
    } \leq \frac1{50}.
\end{eqn}
This contradicts the conclusion that $D_{TV}(L_1,L_2)\geq1/9$ so we conclude as desired.
\end{proof}

With the lemma in hand, we prove the $\Omega(nk)$ lower bound for approximating the cost of KKMC.

\begin{prop}\label{prop:kkmc-cost-hardness}
Let $k^2 = O(n)$. Suppose there exists an algorithm which, with probability at least $2/3$ over its random coin tosses and random inner product matrix $\bfK\sim\mu_{\rank}(n,k)$, correctly computes the optimal cost of the kernel $k$-means clustering instance up to a $(1\pm1/2)$ relative error. Furthermore, suppose that the algorithm reads at most $r$ positions of $\bfK$, possibly adaptively. Then, $r = \Omega(nk)$.
\end{prop}
\begin{proof}
Assume that $nk = o(n^2)$, since otherwise the lower bound is best possible. As in the proof of Theorem \ref{thm:krr-hardness-main}, we have by Chernoff bounds that the number of $\bfe_j$ drawn is $(1\pm\tfrac1{100})n/k$ with probability at least $99/100$ for $n/k$ large enough for all $j\in[k]$ simultaneously. Note then that the optimal cost when $\bfx\sim\Unif(\{\bfe_j:j\in[k]\})$ is $0$, since we can just take the centers to each be $\bfe_j$ for $j\in[k]$. On the other hand,  when $\bfx=\bfe_{k+1}$, then there are more than $k$ types of vectors and thus the cost cannot be $0$.

Thus, with probability $99/100$, the algorithm must correctly distinguish the two cases whenever the algorithm correctly approximates the optimal cost up to $(1\pm1/2)$ relative error. Since the algorithm does this with probability $2/3$, by the union bound, the overall algorithm succeeds with probability at least $1-(\tfrac1{100}+\tfrac13)\geq 5/8$. Thus, the algorithm reads $\Omega(nk)$ kernel entries by Lemma \ref{lem:rank-hardness}.
\end{proof}

\subsection{Main lower bound}
\subsubsection{The construction}
We describe our hard input distribution $\mu_\KKMC(n,k,\eps)$, formed as an inner product matrix of points drawn from the ambient space $\mathbb R^{k/\eps}$.

\begin{dfn}[Hard input distribution -- kernel $k$-means clustering]\label{dfn:kkmc-hard-dist}
Let $\eps>0,k,n$ be such that $k\binom{\eps^{-1}}2 = o(n)$ and $k/\eps = \omega(1)$. We first define a distribution $\nu_\KKMC(k,\eps)$ over vectors in $\mathbb R^{k/\eps}$ as follows. First divide the $k/\eps$ coordinates into $k$ blocks of $1/\eps$ dimensions. Then, we sample our point set as follows: first uniformly select some block $j\in[k]$, and then uniformly select one of the $\binom{1/\eps}{2}$ pairs $(j_1,j_2)$ where $j_1,j_2\in[1/\eps]$ with $j_1\neq j_2$, and then output $\bfv_{j,j_1,j_2}\coloneqq(\bfe_{\ell_1} + \bfe_{\ell_2})/\sqrt2$, where $\ell_1 = j/\eps + j_1, \ell_2 = j/\eps + j_2$. We then generate an i.i.d.\ sample $\{\bfx_i\}_{i=1}^n$ of $n$ points drawn from $\nu_\KKMC(k,\eps)$ and then generate $\bfK\sim\mu_\KKMC(n,k,\eps)$ by setting it to be the inner product matrix of $\{\bfx_i\}_{i=1}^n$, i.e.\ $\bfe_i^\top\bfK\bfe_j\coloneqq \bfx_i\cdot\bfx_j$. For $\bfx$ in the support of $\nu_\KKMC(k,\eps)$, we let $\block(\bfx)$ denote the $j\in[k]$ such that $\bfx = \bfv_{j,j_1,j_2}$.
\end{dfn}

Intuitively, we are adding ``edges'' between pairs of coordinates in the same block of $1/\eps$ coordinates, so that clusterings that associate points in the same block together have lower cost.

In this section, we will prove the following main theorem:
\begin{thm}[Query lower bound for kernel $k$-means clustering]\label{thm:main-lower-bound}
    Let $\eps,k,n$ be such that $k\binom{\eps^{-1}}2 = o(n)$. Suppose an algorithm $\mathcal A$ finds a $(1\pm\eps)$-approximate solution to a kernel $k$-means clustering instance drawn from $\mu_\KKMC(n,k,\eps)$ with probability at least $2/3$ over its random coin tosses and the input distribution. Then, $\mathcal A$ makes at least $\Omega(nk/\eps)$ kernel queries.
\end{thm}
This lower bound is tight up to logarithmic factors, nearly matching for example the ridge leverage score algorithm of Theorem 16 in \cite{musco2017recursive} which reads $O\parens*{\tfrac{nk}{\eps}\log\tfrac{k}{\eps}}$ kernel entries.

\subsubsection{Proof overview}
In our proof, we will think of any clustering as being divided into two groups: the points $S$, which are clustered to ``dense'' clusters, and the points $\overline S$, which are clustered to ``sparse'' clusters. Roughly, if we fix a point in a dense cluster and sample points randomly from that cluster, then we have a high probability (at least $\Omega(\eps)$) of finding a point that has nonzero inner product with it. We then argue that if there are not enough points in dense clusters, then the cost of the clustering is too large, so the clustering cannot be a $(1+\eps)$-approximate solution to the optimal kernel $k$-means clustering solution. Then, we show that finding a lot of points in dense clusters can solve a computational problem that requires $\Omega(nk/\eps)$ kernel queries, which then yields Theorem \ref{thm:main-lower-bound}.

The main work that needs to be done is lower bounding the cost of clustering the points that belong to dense clusters, since it is easy to see that sparse clusters have high cost. Among the dense clusters, if the size of the cluster is at least $n/k$, which we call the ``large'' clusters, then the cost is easy to bound. The worrisome part is the ``small'' clusters, which have the potential of having very small cost per point. We will show that if we carefully bound the cost of small clusters as a function of their size, then if we don't have too many points total that belong to small clusters, then the cost is still high enough to achieve the desired result.

\subsection{Cost computations}
\subsubsection{The cost of a good clustering}
Consider the clustering that assigns all points supported in the same block with each other. We first do our cost computations for the average case, where every vector $\bfv_{j,j_1,j_2}$ is drawn the same number of times. Then, the first block has center
\begin{eqn}
    \frac1{\binom{\eps^{-1}}2}\sum_{(i,j)\in\binom{[\eps^{-1}]}2}\frac{\bfe_i + \bfe_j}{\sqrt2} = \frac{(\eps^{-1}-1)}{\sqrt2\binom{\eps^{-1}}2}\sum_{i\in[\eps^{-1}]}\bfe_i = \sqrt2\eps\sum_{i\in[\eps^{-1}]}\bfe_i
\end{eqn}
and the center for the rest of the blocks is similar. Then, the cost of a single point $(\bfe_{i^*} + \bfe_{j^*})/\sqrt2$ is
\begin{eqn}
    \norm*{\frac{\bfe_{i^*} + \bfe_{j^*}}{\sqrt2} - \sqrt2\eps\sum_{i\in[\eps^{-1}]}\bfe_i}_2^2 = 2\parens*{\frac1{\sqrt2} - \sqrt2\eps}^2 + \parens*{\eps^{-1}-2}\parens*{\sqrt2\eps}^2 = 1 - 4\eps + \eps^{-1}2\eps^2 - 4\eps^2 = 1 - 2\eps - 4\eps^2.
\end{eqn}
Thus, the cost of this clustering is like $n(1-2\eps)$. Note that this computation also gives a lower bound on the cost of a cluster containing $n/k$ points, since for any cluster of size $n/k$, we can clearly improve its cost while we can swap points to be supported on the same block.

Now by Chernoff bounds, with probability tending to $1$ as $n/k\binom{\eps^{-1}}2$ tends to infinity, the cost of this clustering is bounded above by
\begin{eqn}
    n\parens*{1 - \parens*{1-\frac1{40}}2 \eps} = n\parens*{1-\frac{79}{40}\eps}.
\end{eqn}
and the cost of any cluster of size $n/k$ is bounded below by
\begin{eqn}
    \frac{n}{k}\parens*{1-\parens*{1+\frac1{40}}2\eps} = \frac{n}{k}\parens*{1-\frac{81}{40}\eps}.
\end{eqn}

This proves the following lemmas.
\begin{lem}[Cost bound for an optimal clustering]\label{lem:cost-optimal}
With probability at least $99/100$, the cost of an optimal clustering is at most $n(1-(79/40)\eps)$.
\end{lem}
\begin{lem}[Cost bound for a large cluster]\label{lem:cost-large-cluster}
Let $C$ be a cluster of size at least $n/k$. Then with probability at least $99/100$, the cost per point of $C$ is bounded below by $1-(81/40)\eps$.
\end{lem}

\subsubsection{The cost of a small cluster}
We will first prove a lower bound on the cost of a fixed cluster $C$. In this section, we will parameterize our lower bound only by the size of the cluster, and then later use this result to lower bound the cost of any clustering. When we prove our lower bound result, it will be useful to introduce the quantities
\begin{eqn}
    \alpha\coloneqq\frac{n}{k\binom{\eps^{-1}}2},\qquad \tau\coloneqq\frac{\abs{C}}\alpha.
\end{eqn}
Intuitively, $\alpha$ is the number of copies of a vector $\bfv_{j,j_1,j_2}$ we expect to draw, and $\tau$ is the minimal number of different types of vectors $\bfv_{j,j_1,j_2}$ we can have in our cluster $C$. Note that by the union and Chernoff bounds, there are $(1\pm\gamma)\alpha$ copies of each vector with probability $k\binom{\eps^{-1}}2\exp\parens*{-\gamma^2\alpha/3}$, where $\gamma$ is a small constant to be chosen. The lower bound we prove then is the following.

\begin{lem}[Cost bound for a small cluster]\label{lem:cost-fixed-cluster}
    Let $\gamma$ be a constant small enough so that
    \begin{eqn}\label{eqn:gamma-choice}
        (1+\gamma)^2\frac{(1+2\sqrt\gamma)^2}{(1-2\sqrt\gamma)^3}\leq 1+\frac1{20}.
    \end{eqn}
    Let $\alpha$ and $\tau$ be defined as above, with respect to a cluster $C$ of points drawn from $\nu_\KKMC(k,\eps)$ with size bounded by
    \begin{eqn}\label{eqn:cluster-size-bound}
        \frac{\alpha}{\gamma} = \Theta(\alpha)\leq \abs{C}\leq \frac{n}{k}
    \end{eqn}
    Additionally, define the quantity
    \begin{eqn}\label{eqn:def-kappa}
        \kappa\coloneqq (\eps^{-1}-1/2) - \sqrt{(\eps^{-1}-1/2)^2 - 2\tau}.
    \end{eqn}
    Then, the cost on this cluster is bounded below by
    \begin{eqn}
    \cost(C)\geq \abs{C} - \parens*{1+\frac1{20}}\frac{\alpha}{2}\frac{\kappa(\eps^{-1}-1)^2 + \kappa^2(\eps^{-1}-\kappa)}{\tau}.
    \end{eqn}
\end{lem}

\begin{rem}
Note that if a cluster has size at most $\alpha$, then its cost can be $0$ by taking them all to be the same vector, and here, we only require the cluster to have size a constant times $\alpha$ to get the lower bound.
\end{rem}

\begin{proof}
    To prove this lemma, we will first reduce to the case of considering clusters which have all copies of all of its vectors, and then conclude the lower bound by optimizing the cost of clusters of this form.

    \paragraph{Reduction to maximizing a sum of squares.}
    Since we wish to lower bound the cost of this cluster, we assume that it is supported on $\eps^{-1}$ coordinates since using more than $\eps^{-1}$ coordinates is clearly suboptimal, and we make take the entire cluster to be drawn from one block by the upper bound on the size of $C$. Then for $i\in[\eps^{-1}]$, let $n_i$ denote the number of vectors supported on the $i$th coordinate, so $\sum_{i\in[\eps^{-1}]}n_i = 2\abs{C}$. Let $\bfmu$ denote the center of $C$. Then, it's easy to see that
    \begin{eqn}
      \bfmu = \frac1{\abs{C}}\sum_{i\in[\eps^{-1}]} \frac{n_i}{\sqrt2}\bfe_i.
    \end{eqn}
    Thus, for each point $\bfv_{i,j}\coloneqq (\bfe_i+\bfe_j)/\sqrt2$ in $C$, its cost contribution is
    \begin{eqn}
      \parens*{\frac1{\sqrt2}-\bfmu_i}^2 + \parens*{\frac1{\sqrt2}-\bfmu_j}^2 + \sum_{\ell\in[\eps^{-1}]\setminus\{i,j\}}\parens*{0-\bfmu_\ell}^2 = 1 + \norm*{\bfmu}_2^2 - \sqrt2(\bfmu_i+\bfmu_j).
    \end{eqn}
    Then, the total cost for the cluster is
    \begin{eqn}
      \sum_{\bfv_{i,j}\in C}1+\norm{\bfmu}_2^2 - \sqrt2(\bfmu_i+\bfmu_j) &= \sum_{\bfv_{i,j}\in C}1+\norm{\bfmu}_2^2 - \sqrt2\parens*{\frac{n_i}{\sqrt2\abs{C}}+\frac{n_j}{\sqrt2\abs{C}}} \\
      &= \abs{C}(1+\norm{\bfmu}_2^2) - \frac1{\abs{C}}\sum_{\bfv_{i,j}\in C} n_i + n_j \\
      &= \abs{C}\parens*{1+\sum_{i\in[\eps^{-1}]}\frac{n_i^2}{2\abs{C}^2}} -  \frac1{\abs{C}}\sum_{i\in[\eps^{-1}]} n_i^2 \\
      &= \abs{C} -  \frac1{2\abs{C}}\sum_{i\in[\eps^{-1}]} n_i^2.
    \end{eqn}
    Thus, going forward, we will forget about the $\abs{C}$ and focus only on the sum of squared counts. In the following arguments, we may give the impression of adding vectors to the cluster to bound this quantity above, but we will do this without touching this $\abs{C}$ term.

    \paragraph{Discretization.}
    Consider two vectors $\bfv_{i_1,i_2},\bfv_{j_1,j_2}$ where $n_{i_1}+n_{i_2}\geq n_{j_1}+n_{j_2}$. Then, we have that
    \begin{eqn}
        (n_{i_1}+1)^2 + (n_{i_2}+1)^2 + (n_{j_1}-1)^2 + (n_{j_2}-1)^2 &= n_{i_1}^2 + n_{i_2}^2 + n_{j_1}^2 + n_{j_2}^2 + 4 + 2(n_{i_1}+n_{i_2}-n_{j_1}-n_{j_2}) \\
        &> n_{i_1}^2 + n_{i_2}^2 + n_{j_1}^2 + n_{j_2}^2
    \end{eqn}
    so we can strictly improve the cost of any clustering by changing a vector with low total count of coordinates to one with higher total count of coordinates. Thus, we may reduce our lower bound proof to the case where every vector type $\bfv_{i,j}$ has every copy of itself, except for possibly one ``leftover'' vector type. By adding in all copies of this leftover vector type, which only makes the sum of squared counts larger, we assume that every vector type has every copy of itself.

    \paragraph{Filling up coordinates.}
    We do a similar argument as in the above to further constrain the form of the clusters we consider. Note that
    \begin{eqn}
        (n_i+1)^2 + (n_j-1)^2 = n_i^2 + n_j^2 + 2 + 2(n_i-n_j) > n_i^2 + n_j^2
    \end{eqn}
    for $n_i\geq n_j$, so we improve the clustering by iteratively swapping vectors $v_{i,j}$ to $v_{i^*,j}$ where $i^*$ is the coordinate with the largest count $n_{i^*}$ that still hasn't exhausted the $(\eps^{-1}-1)$ different types of vectors $v_{i^*,j}$ for $j\in[\eps^{-1}]\setminus\{j\}$ that is supported on the $i^*$th coordinate. By relabeling the coordinates if necessary, WLOG assume that we fill up the coordinates $i=1,2,\dots,\eps^{-1}$ in order.

    On the first coordinate, we can use $(\eps^{-1}-1)$ types of vectors to fill it all the way up to
    \begin{eqn}
      n_1 = (\eps^{-1}-1)(1\pm\gamma)\alpha = (1\pm\gamma)(\eps^{-1}-1)\frac{n}{k\binom{\eps^{-1}}{2}} = (1\pm\gamma)\frac{2n\eps}{k}
    \end{eqn}
    vectors, where the $(\eps^{-1}-1)$ types of vectors are $(\bfe_1+\bfe_j)/\sqrt2$ for $j\in[\eps^{-1}]\setminus\{1\}$. Note that at this point, every other coordinate will have $(1\pm\gamma)\alpha$ vectors. Then we need to maximize the second coordinate, which already has $1$ type and $(1\pm\gamma)\alpha$ vectors, so we can use $(\eps^{-1}-2)$ additional types of vectors to fill it all the way up, at which point all the other coordinates will have $2(1\pm\gamma)\alpha$ vectors. We can do this until we fill up $K$ coordinates at which point we have used $T$ types of vectors. Note that we can solve for $K$ by solving
    \begin{eqn}
      T = \sum_{i=1}^K \eps^{-1}-i = K\eps^{-1} - \frac{K(K+1)}{2} &\iff \frac12K^2 + \parens*{\frac12 - \eps^{-1}}K + T = 0 \\
      &\iff K = (\eps^{-1}-1/2) \pm \sqrt{(\eps^{-1}-1/2)^2 - 2T}.
    \end{eqn}
    Since we can only have $\eps^{-1}$ coordinates, we need to take the solution
    \begin{eqn}
      K = (\eps^{-1}-1/2) - \sqrt{(\eps^{-1}-1/2)^2 - 2T}.
    \end{eqn}

    \paragraph{Bounding.}
    Note that from filling up the last leftover vector type, we have that $T\leq \ceil{\abs{C}/(1-\gamma)\alpha}$ and $T\geq \floor{\abs{C}/(1+\gamma)\alpha}$. Thus, $\abs{C} = (1\pm\gamma)\alpha (T\pm 1)$. Now note that we require that $\alpha/\gamma\leq \abs{C}$ in equation (\ref{eqn:cluster-size-bound}), which means that
    \begin{eqn}
        \tau = \frac{\abs{C}}{\alpha}\geq \frac{\alpha/\gamma}{\alpha} = \frac1\gamma\implies \tau\gamma\geq 1
    \end{eqn}
    and thus $T = (1\pm\gamma)\tau\pm 1 = (1\pm2\gamma)\tau$. It then follows that for $\kappa$ defined as in (\ref{eqn:def-kappa}),
    \begin{eqn}
        K &= (\eps^{-1}-1/2) - \sqrt{(\eps^{-1}-1/2)^2 - 2(1\pm2\gamma)\tau} \\
        &= \bracks*{(\eps^{-1}-1/2) - \sqrt{(\eps^{-1}-1/2)^2 - 2\tau}} \pm \sqrt{4\gamma\tau} \\
        &= \kappa \pm \sqrt{4\gamma\binom{\kappa}2} = (1\pm2\sqrt{\gamma})\kappa
    \end{eqn}
    For $K$ coordinates we fill the coordinates all the way up to $(\eps^{-1}-1)(1\pm\gamma)\alpha$ and otherwise we have $K(1\pm\gamma)\alpha$ vectors. Then, the sum of squared counts is bounded by
    \begin{eqn}
      \sum_{i\in[\eps^{-1}]}n_i^2 &\leq K \parens*{(\eps^{-1}-1)(1+\gamma)\alpha}^2 + (\eps^{-1}-K)\parens*{K (1+\gamma)\alpha}^2 \\
      &= ((1+\gamma)\alpha)^2\bracks*{K(\eps^{-1}-1)^2 + K^2(\eps^{-1}-K)}
    \end{eqn}
    so the cost on this cluster is at least
    \begin{eqn}
        \abs{C} - \frac1{2\abs{C}}\sum_{i\in[\eps^{-1}]}n_i^2 &\geq \abs{C} - \frac1{2\abs{C}}((1+\gamma)\alpha)^2\bracks*{K(\eps^{-1}-1)^2 + K^2(\eps^{-1}-K)} \\
        &= \abs{C} - \parens*{1+\frac1{20}}\frac{\alpha}{2}\frac{\kappa(\eps^{-1}-1)^2 + \kappa^2(\eps^{-1}-\kappa)}{\tau}
    \end{eqn}
    since we chose $\gamma$ to be small enough in equation (\ref{eqn:gamma-choice}).
\end{proof}

\subsubsection{Optimizing over \texorpdfstring{$k$}{k} clusters}
In Lemma \ref{lem:cost-fixed-cluster}, we have found a lower bound on the cost of a fixed cluster that only makes reference to the size of the cluster. All that is left to do is to optimize the sum of these functions under the constraint of the total number of points to cluster. Recall from the proof overview that we are in the context of lower bounding the cost of a subset of the points $S$, which we wish to show must be large. Thus, we will assume that $\abs{S}\leq 2n/5$ in this lower bound. The formal statement of the lemma we prove here is the following:

\begin{lem}[Cost bound for $\ell$ clusters]\label{lem:cost-bound-l-clusters}
    Suppose $S$ is a set of at most $\abs{S}\leq 2n/5$ points drawn from $\nu_\KKMC(k,\eps)$. Then, for any clustering $\mathcal C_S$ of $S$ into $\ell\leq k$ clusters,
    \begin{eqn}
        \cost(\mathcal C_S)\geq\abs{S} - \frac{77}{40}n\eps.
    \end{eqn}
\end{lem}
\begin{proof}
    Note that we only decrease the cost of a clustering if we allow a single cluster to be split up into multiple clusters. Then for any ``large'' cluster $C$ with size at least $\abs{C} > n/k$, we can treat every $n/k$ points as a separate cluster and lower bound its cost at $1-2\eps$ per point. Then, let $S = \mathcal S\cup \mathcal T$, where $\mathcal S$ is the set of clusters that now belong to clusters of size at most $n/k$ and $\mathcal T$ is the set of points whose cost we have bounded below by $1-(81/40)\eps$ by Lemma \ref{lem:cost-large-cluster}.

    Recall that our lower bound for a single cluster, Lemma \ref{lem:cost-fixed-cluster}, only applies to clusters of size at least $\alpha/\gamma$. Let $\mathcal L\subseteq[\ell]$ be the set of indices of clusters such that the lower bound applies and let $\mathcal M\subseteq \mathcal S$ be the set of all points in such a cluster. Now applying this lower bound for these clusters and lower bounding by $0$ for the others, we arrive at the lower bound
    \begin{eqn}
        \cost(\mathcal C_S) \geq \abs{\mathcal T}\parens*{1-\frac{81}{40}\eps} + \sum_{i\in\mathcal L}\abs{C_i} - \parens*{1+\frac1{20}}\frac{\alpha}{2}\frac{\kappa_i(\eps^{-1}-1)^2 + \kappa_i^2(\eps^{-1}-\kappa_i)}{\tau_i}.
    \end{eqn}
    Together with the constraint that there are $\abs{\mathcal M}$ points for which we can apply the lower bound of Lemma \ref{lem:cost-fixed-cluster}, we now focus on the optimization problem
    \begin{eqn}\label{eqn:main-opt-prob}
        \text{minimize} &&& \abs{\mathcal M} - \parens*{1+\frac1{20}}\frac{\alpha}{2}\sum_{i\in\mathcal L} \frac{\kappa_i(\eps^{-1}-1)^2 + \kappa_i^2(\eps^{-1}-\kappa_i)}{\tau_i} \\
        \text{subject to} &&& \alpha\sum_{i\in\mathcal L} \tau_i = \abs{\mathcal M},\qquad 0\leq \tau_i\leq \binom{\eps^{-1}}{2} = \frac{(\eps^{-1}-1/2)^2}{2} - \frac18
    \end{eqn}
    where
    \begin{eqn}
        \kappa_i = (\eps^{-1}-1/2) - \sqrt{(\eps^{-1}-1/2)^2 - 2\tau_i}.
    \end{eqn}
    We now introduce less cumbersome notation to get our bound. Let
    \begin{eqn}
        R\coloneqq (\eps^{-1}-1/2),\qquad u_i\coloneqq R^2 - 2\tau_i.
    \end{eqn}
    Then by plugging in definitions, the optimization problem is now
    \begin{eqn}
        \text{minimize} &&& \abs{\mathcal M} - \parens*{1+\frac1{20}}\alpha\sum_{i\in\mathcal L} \frac{(R-\sqrt{u_i})(R-1/2)^2 + (R-\sqrt{u_i})^2(\sqrt{u_i}+1/2)}{R^2 - u_i} \\
        \text{subject to} &&& \sum_{i\in\mathcal L}u_i = \abs{\mathcal L}R^2 - 2\abs{\mathcal M}/\alpha,\qquad \frac14\leq u_i\leq R^2
    \end{eqn}
    Now note (\href{http://www.wolframalpha.com/input/?i=((R-u)(R-1\%2F2)\%5E2+\%2B+(R-u)\%5E2(u\%2B1\%2F2))\%2F(R\%5E2-u\%5E2)}{WolframAlpha link}) that
    \begin{eqn}
        \frac{(R-\sqrt{u_i})(R-1/2)^2 + (R-\sqrt{u_i})^2(\sqrt{u_i}+1/2)}{R^2 - u_i} &= 2R - \frac12 - \parens*{\frac{4R^2-1}{4(R+\sqrt{u_i})} + \sqrt{u_i}} \\
        &\leq 2R - \frac12 - \parens*{\frac{4R^2-1}{8R} + \sqrt{u_i}} = \frac32 R - \frac12 + \frac1{8R} -\sqrt{u_i}.
    \end{eqn}
    Then, by noting that $\sqrt{a+b}\leq \sqrt{a}+\sqrt{b}$, we can optimize
    \begin{eqn}\label{eqn:reduced-opt-prob}
        \text{maximize} &&& \sum_{i\in\mathcal L}-\sqrt{u_i} \\
        \text{subject to} &&& \sum_{i\in\mathcal L}u_i = \abs{\mathcal L}R^2 - 2\abs{\mathcal M}/\alpha,\qquad \frac14\leq u_i\leq R^2
    \end{eqn}
    by setting $u_i = R^2$ for $(\abs{\mathcal L}R^2 - 2\abs{\mathcal M}/\alpha)/R^2$ of the $i$ and the rest to $0$ (the minimum $u_i$ is $1/4$, but allowing it to be $0$ is a relaxation since we optimize over a larger domain). Thus, the value of optimization problem (\ref{eqn:reduced-opt-prob}) is at most
    \begin{eqn}
        \frac{\abs{\mathcal L}R^2 - 2\abs{\mathcal M}/\alpha}{R^2}(-R) = -\parens*{\abs{\mathcal L}R - \frac{2\abs{\mathcal M}}{\alpha R}}.
    \end{eqn}
    Thus, we have that
    \begin{eqn}
        \alpha\parens*{\sum_{i\in\mathcal L} \frac32 R - \frac12 + \frac1{8R} -\sqrt{u_i}}\leq \alpha \abs{\mathcal L}\parens*{\frac32 R - \frac12 + \frac1{8R}} - \alpha\parens*{\abs{\mathcal L}R - \frac{2\abs{\mathcal M}}{\alpha R}}.
    \end{eqn}
    Now recall that we set $\abs{\mathcal M}+\abs{\mathcal T}\leq \abs{S}\leq 2n/5$ and $\abs{\mathcal L}\leq k$ so the above is bounded above by
    \begin{eqn}
        \alpha \abs{\mathcal L}\parens*{\frac32 R - \frac12 + \frac1{8R}} - \alpha\parens*{\abs{\mathcal L}R - \frac{2\abs{\mathcal M}}{\alpha R}} &= \frac12\alpha \abs{\mathcal L}R - \alpha \abs{\mathcal L}\parens*{\frac12 - \frac1{8R}} + \alpha \frac{2\abs{\mathcal M}}{\alpha R} \\
        &\leq \frac12\alpha kR - \alpha \abs{\mathcal L}\parens*{\frac12 - \frac1{8R}} + \frac25 \alpha kR - \frac{2\abs{\mathcal T}}{R} \\
        &= \frac{9}{10}\alpha kR - \alpha \abs{\mathcal L}\parens*{\frac12 - \frac1{8R}} - \frac{2\abs{\mathcal T}}{R} \\
        &\leq \frac9{10}\alpha kR - 2\eps\abs{\mathcal T}.
    \end{eqn}
    Thus, we may lower bound the value of the optimization problem of (\ref{eqn:main-opt-prob}) by
    \begin{eqn}\label{eqn:bound-main-opt-prob}
        \abs*{\mathcal M} - \parens*{1+\frac1{20}}\parens*{\frac9{10}\alpha kR - 2\eps\abs{\mathcal T}}.
    \end{eqn}

    Now finally, note that there are at most $k$ clusters where the lower bound of Lemma \ref{lem:cost-fixed-cluster} doesn't apply, so there are at most
    \begin{eqn}\label{eqn:num-ignored-points}
        \abs{\mathcal S\setminus\mathcal M}\leq k\cdot\frac{\alpha}{\gamma} = \Theta(n\eps^2)\leq \frac1{100}n\eps
    \end{eqn}
    points that we have ignored the cost for, for $\eps$ smaller than some constant. Collecting our bounds of (\ref{eqn:bound-main-opt-prob}) and (\ref{eqn:num-ignored-points}) and plugging our definitions of $\alpha$ and $R$ back in, we obtain a lower bound of
    \begin{eqn}
        \cost(\mathcal C_S) &\geq \abs{\mathcal T}\parens*{1-\frac{81}{40}\eps} + \abs{\mathcal M} - \parens*{1+\frac1{20}}\frac{9}{10}\alpha kR + \parens*{1+\frac1{20}}2\eps\abs{\mathcal T} \\
        &= \parens*{\abs{\mathcal T}+\abs{\mathcal M}}+\parens*{\frac{42}{20}-\frac{81}{40}}\eps\abs{\mathcal T} - \parens*{1+\frac1{20}}\frac{9}{10}\alpha kR \\
        &\geq \parens*{\abs{S}-\abs{\mathcal S\setminus\mathcal M}} - \parens*{1+\frac1{20}}\frac{9}{10}\alpha kR \\
        &\geq \abs{S} - \frac1{100}n\eps  - \parens*{1+\frac1{20}}\frac{9}{5}n\eps > \abs{S} - \frac{77}{40} n\eps
    \end{eqn}
    on the cost of the clustering, as desired.
\end{proof}

\subsection{Hardness}

\subsubsection{Cost lemma}
We now prove a lemma that translates our cost computations from the above section into a statement about the probability of sampling nonzero inner products. Intuitively, we will prove that an approximately optimal solution to the kernel $k$-means clustering instance must output a clustering such that at least $2n/5$ of the points belong to a cluster with lots of points that share a coordinate with the point, i.e.\ ``neighbors''.

\begin{lem}[Sampling probability of an approximate solution]\label{lem:cost-to-sampling}
Suppose that $\mathcal C$ is a $(1+\eps/40)$-approximate solution to a kernel $k$-means clustering instance drawn as $\bfK\sim\mu_\KKMC(n,k,\eps)$. Then for at least $2n/5$ of the points, if we uniformly sample dot products between the point and other points in its cluster, then there is at least an $\eps/80$ probability of sampling a point that has nonzero inner product with the point.
\end{lem}
\begin{proof}
    Suppose for contradiction that there are at most $2n/5$ points belonging to a cluster such that sampling uniformly from the cluster yields at least an $\eps/80$ probability of sampling a point that has nonzero inner product with that point, which we refer to as a \emph{neighbor}. Let $S$ be the set of points that belong to such a cluster with at least probability $\eps/80$ of sampling a neighbor, and let $\overline S$ be the complement. We first compute the cost of a point $(\bfe_i+\bfe_j)/\sqrt2$ in $\overline S$. Let $C$ be the point's cluster and let $n_i, n_j$ be the number of points in the cluster that has support on the $i$th coordinate. Then, $n_i/\abs{C}$ and $n_j/\abs{C}$ are both at most $\eps/80$. Now note that the $i$ and $j$th coordinates of the center are $n_i/(\sqrt2 \abs{C})$ and $n_j/(\sqrt2\abs{C})$, so the cost of that point is at least
    \begin{eqn}
        \parens*{\frac1{\sqrt2} - \frac1{\sqrt2}\frac{n_i}{\abs{C}}}^2 + \parens*{\frac1{\sqrt2} - \frac1{\sqrt2}\frac{n_j}{\abs{C}}}^2\geq 1 - \frac1{40}\eps.
    \end{eqn}
    Then $\abs{S}\leq 2n/5$, so we may use the bounds from Lemma \ref{lem:cost-bound-l-clusters} to bound the cost from below by
    \begin{eqn}
        \abs*{\overline S}\parens*{1-\frac1{40}\eps} + \abs{S} - \frac{77}{40} n\eps\geq  n\parens*{1-\frac{78}{40}\eps}.
    \end{eqn}
    Now recall that by Lemma \ref{lem:cost-optimal}, the optimal solution has cost at most $n(1-(79/40)\eps)$, so a $(1+\eps/40)$-approximate solution needs to have cost at most
    \begin{eqn}
        n\parens*{1-\frac{79}{40}\eps}\parens*{1+\frac1{40}\eps} = n\parens*{1-\frac{78}{40}\eps - \frac{79}{1600}\eps^2} < n\parens*{1-\frac{78}{40}\eps}
    \end{eqn}
    which the above solution does not.
\end{proof}

\subsubsection{Hardness reduction}
Finally, we give the hardness result. Consider the following computational problem \textsc{LabelKKMC}. Recall the definition of $\nu_\KKMC$ and $\block$ from Definition \ref{dfn:kkmc-hard-dist}.

\begin{dfn}[\textsc{LabelKKMC}]
  We first sample $n$ points $\{\bfx_i\}_{i=1}^n$ from our hard point set $\nu_\KKMC(k,\eps)$, label the identity of the first $n/2$ points, and then ask an algorithm to correctly give $\block(\bfx_i)$ for $1/6$ of the remaining $n/2$ points.
\end{dfn}

We will show that this problem requires reading $\Omega(nk/\eps)$ kernel entries and that an algorithm solving the KKMC problem on this instance can solve this problem. We first prove the following lemma:

\begin{dfn}[\textsc{LabelSingleKKMC}]\label{dfn:label-single-kkmc}
  Given as input $\bfx\sim\nu_{KKMC}(k,\eps)$, determine $\block(\bfx)$.
\end{dfn}

\begin{lem}[Hardness of \textsc{LabelSingleKKMC}]\label{lem:labelkkmc-single-draw}
  Let $\log_2 k\geq 12$. Suppose there exists an algorithm $\mathcal A$, possibly randomized, that correctly solves \textsc{LabelSingleKKMC} with probability at least $1/100$ over $\bfx\sim\nu_\KKMC(k,\eps)$ and its random coin tosses. Furthermore, suppose that $\mathcal A$ makes at most $q$ inner product queries of the form $\bfx\cdot\bfv_{\ell,\ell_1,\ell_2}$, possibly adaptively, on any input. Then, $q\geq (k/\eps)/100$.
\end{lem}
\begin{proof}
By way of Yao's minimax principle \cite{yao1977probabilistic}, assume that the algorithm is deterministic. Note first that for a single $\bfv_{j,j_1,j_2}$, its inner product with any vector drawn from a different block is $0$. Within the same block, its inner product with another $\bfv_{j,j_1',j_2'}$ is $0$ as well, unless $j_1 = j_1'$ or $j_2 = j_2'$. This happens with probability
\begin{eqn}
    1 - \frac{\binom{\eps^{-1}-2}2}{\binom{\eps^{-1}}2} = \frac{\eps^{-1}(\eps^{-1}-1)-(\eps^{-1}-2)(\eps^{-1}-3)}{\eps^{-1}(\eps^{-1}-1)} = \frac{4\eps^{-1}-6}{\eps^{-1}(\eps^{-1}-1)} = \frac{4\eps-6\eps^2}{1-\eps}\leq 8\eps
\end{eqn}
for $\eps\leq 1/2$ and thus for a fixed $\bfv_{\ell,\ell_1,\ell_2}$, we have that
\begin{eqn}\label{eqn:nu-kkmc-prob}
  \Pr_{\bfx\sim\nu_\KKMC(k,\eps)}(\bfx\cdot \bfv_{\ell,\ell_1,\ell_2}\neq 0)\leq \frac{8\eps}{k}.
\end{eqn}

We now use the above to first show a lower bound of $q = \Omega(k/\eps)$ on the number of adaptive inner product queries $\bfx\cdot\bfv_{\ell,\ell_1,\ell_2}$ required to find $\block(\bfx)$ for a single draw $\bfx\sim\nu_\KKMC(k,\eps)$.

Assume for simplicity that $k$ is a power of $2$, and for each $m\in[\log_2 k]$, consider the hypothesis test $\mathcal H_m$ of deciding whether the $m$th bit of $\block(\bfx)\in[k]$ is $0$ or $1$. Note that since we choose $j\in[k]$ uniformly, each of the $\log_2 k$ hypothesis tests are independent and identical and thus the error probability of the hypothesis test of the optimal success probability is the same for each hypothesis test. Note that an algorithm succeeds in correctly outputting $\block(\bfx)$ if and only if it succeeds on all $\log_2 k$ of the hypothesis tests. If the optimal success probability is at most $2/3$, then for $\log_2 k\geq 12$, we have that the success probability on all $\log_2 k$ of the hypothesis tests is at most
\begin{eqn}
  \parens*{\frac23}^{\log_2 k}\leq \parens*{\frac23}^{12} < \frac1{100}
\end{eqn}
which means it does not have the required guarantees. Thus, $\mathcal A$ must solve each hypothesis test $\mathcal H_m$ with probability at least $2/3$.

Now fix a hypothesis test $\mathcal H_m$ from the above, let $\mathcal E_0$ be the event that the $m$th bit of $\block(\bfx)$ is $0$, and let $\mathcal E_1$ be the event that the $m$th bit of $\block(\bfx)$ is $1$. Let $S$ be the random variable indicating the list of positions of $\bfK$ read by $\mathcal A$ and its corresponding values, let $L_0$ denote the distribution of $S$ conditioned on $\mathcal E_0$, and let $L_1$ denote the distribution of $S$ conditioned on $\mathcal E_1$. Then by Proposition 2.58 of \cite{bar2002complexity}, we have that
\begin{eqn}
  \frac{1+D_{TV}(L_0,L_1)}2\geq\frac23
\end{eqn}
so $D_{TV}(L_0,L_1)\geq1/3$.

Now suppose for contradiction that $\mathcal A$ makes $q\leq (k/\eps)/100$ queries on any input. Since $\mathcal A$ is deterministic, it makes the same sequence of inner product queries if it reads a sequence of $q$ zeros. For a fixed query $\bfv_{\ell,\ell_1,\ell_2}$, the probability that $\bfx\cdot \bfv_{\ell,\ell_1,\ell_2}\neq 0$ is at most $8\eps/k$ by equation (\ref{eqn:nu-kkmc-prob}). Then by a union bound over the $q$ queries, the probability that the algorithm seems any zeros is at most $8(\eps/k)q\leq 8/100$. Then, letting $\Omega$ denote the support of $S$ and $s_0$ the value of $S$ when $\mathcal A$ reads all zeros, we have that
\begin{eqn}
  D_{TV}(L_0,L_1) &= \abs*{\Pr\parens*{S=s_0\mid\mathcal E_0}-\Pr\parens*{S=s_0\mid\mathcal E_1}} + \sum_{s\in\Omega\setminus\{s_0\}}\abs*{\Pr\parens*{S=s\mid\mathcal E_0}-\Pr\parens*{S=s\mid\mathcal E_1}} \\
  &\leq \frac8{100} + \Pr\parens*{S\neq s_0\mid\mathcal E_0} + \Pr\parens*{S\neq s_0\mid\mathcal E_1}\leq \frac{24}{100} < \frac14
\end{eqn}
which is a contradiction. Thus, we conclude that $q > (k/\eps)/100$, as desired.
\end{proof}

Using the above lemma, we may prove the full hardness of \textsc{LabelKKMC}.

\begin{lem}[Hardness of \textsc{LabelKKMC}]
Suppose an algorithm $\mathcal A$, possibly randomized, solves \textsc{LabelKKMC} with probability at least $2/3$ over the input distribution $\nu_\KKMC(k,\eps)$ and the algorithm's random coin tosses. Then, $\mathcal A$ makes $\Omega(nk/\eps)$ kernel queries.
\end{lem}
\begin{proof}
Given such an algorithm, we given an algorithm $\mathcal B$ solving \textsc{LabelSingleKKMC} (see definition \ref{dfn:label-single-kkmc}) as follows. We first generate $n-1$ points $\{\bfx_i\}_{i=1}^{n-1}$ drawn i.i.d.\ from $\nu_{KKMC}$. We then draw a random index $i^*\sim\Unif([n/2])$, set the $i^*+n/2$th point to $\bfx$, and the rest of the points according to $\{\bfx_i\}_{i=1}^{n-1}$. We then run $\mathcal A$ on this input instance as follows. We can clearly give $\mathcal A$ the labels of the first $n/2$ elements of$\{\bfx_i\}_{i=1}^{n-1}$ without making inner product queries to $\bfx$. Whenever we need to read a kernel entry that doesn't involve the $i^*$th element, we generate the inner product without making queries to $\bfx$. Otherwise, we make the required inner product query $\bfx\cdot\bfv_{\ell,\ell_1,\ell_2}$ requested by $\mathcal A$. With probability at least $2/3$, $\mathcal A$ succeeds in outputting $\block(\bfx_i)$ for at least a $1/6$ fraction of the last $n/2$ input points. By symmetry, $\bfx$ has a correct label with probability at least $1/6$ in this event. Then by independence, the algorithm succeeds with probability at least $(2/3)(1/6) = 1/9 > 1/100$. Thus, $\mathcal B$ indeed solves \textsc{LabelSingleKKMC}.

We now bound the number of inner product queries made by $\mathcal B$. Suppose for contradiction that $\mathcal A$ makes at most $o(nk/\eps)$ total kernel queries. Then by averaging, $\mathcal A$ makes at most $(k/\eps)/200$ kernel queries for a $199/200$ fraction of the $n/2$ last rows. Similarly, $\mathcal A$ makes at most $(k/\eps)/200$ kernel queries for a $199/200$ fraction of the last $n/2$ columns. Thus, by a union bound, $\mathcal A$ makes at most $(k/\eps)/100$ inner product queries for a $99/100$ fraction of the last $n/2$ input points. By symmetry, $\mathcal A$ makes at most $(k/\eps)/100$ inner product queries on $\bfx$ with probability at least $99/100$. This contradicts Lemma \ref{lem:labelkkmc-single-draw}. Thus, we conclude that $\mathcal A$ makes at least $\Omega(nk/\eps)$ kernel queries, as desired.
\end{proof}

Finally, we use the above lemma to show the hardness of kernel $k$-means clustering.
\begin{cor}
Suppose an algorithm $\mathcal A$ gives a $(1+\eps/40)$-approximate kernel $k$-means clustering solution with probability at least $2/3$ over the input distribution $\bfK\sim\mu_\KKMC(n,k,\eps)$ and the algorithm's random coin tosses. Then, $\mathcal A$ makes $\Omega(nk/\eps)$ kernel queries.
\end{cor}
\begin{proof}
    Using a $(1+\eps/40)$-approximate algorithm for $k$-means clustering, we can solve the computational problem described above as follows. We first cluster all the points using $\mathcal A$. Then, note that by Lemma \ref{lem:cost-to-sampling}, at least $2/5$ of the points must belong to a cluster such that sampling $O(1/\eps)$ points within its cluster allows us to find a point such that at least one coordinate matches a labeled point's coordinate. Then, on average, we will get $1/5$ of these correct and thus $1/6$ of these with very high probability by Chernoff bounds. Note that this used $Q + O(n/\eps)$ kernel queries, where $Q$ is the number of kernel queries that the kernel $k$-means step used. Then, since $Q + O(n/\eps) = \Omega(nk/\eps)$, we have that $Q = \Omega(nk/\eps)$, as desired.
\end{proof}

Finally, we obtain Theorem \ref{thm:main-lower-bound} by rescaling $\eps$ by a constant factor.

\section{Clustering mixtures of Gaussians}\label{section:mog-upper-bound}
In this section, we show that our worst case kernel query complexity lower bounds for the kernel $k$-means clustering problem are pessimistic by a factor of $k$ when our input instance is mixture of $k$ Gaussians. More specifically, we prove the following theorem:

\begin{thm}[Clustering mixtures of Gaussians]\label{thm:cluster-gaussians-main}
Let $m = \Omega(\eps^{-1}\log n)$ as specified by Corollary \ref{cor:gaussian-sketch-dimension}. Suppose we have a mixture of $k$ Gaussians with isotropic covariance $\sigma^2\bfI_d$ and means $(\bfmu_\ell)_{\ell=1}^k$ in $\mathbb R^d$. Furthermore, suppose that the Gaussian means $\bfmu_\ell$ are all separated by at least $\norm*{\bfmu_{\ell_1}-\bfmu_{\ell_2}}_2\geq \Omega(\sigma\sqrt{\log k})$ as specified by Theorem 5.1 of \cite{regev2017learning} and $\norm*{\bfmu_{\ell_1}-\bfmu_{\ell_2}}_2\geq \Omega(\sigma\sqrt{\log\log n + \log\eps^{-1}})$ as specified by Lemma \ref{lem:distinguish-gaussians} with $\delta = (2m+k)^{-3}$. Finally, suppose that we are in the parameter regime of $\poly(k,1/\eps,d,\log n) = O(\sqrt{n})$, $d\eps\geq 1$, and $k/\eps\leq d\leq n/10$. Then, there exists an algorithm outputting a $(1+O(\eps))$-approximate $k$-means clustering solution with probability at least $2/3$.
\end{thm}

\subsection{Proof overview}
By Theorem 5.1 of \cite{regev2017learning}, we can in $s = \poly(k,1/\eps,d)$ samples compute  approximations $(\hat\bfmu_\ell)_{\ell=1}^k$ to the true Gaussian means $(\bfmu_\ell)_{\ell=1}^k$ up to
\begin{eqn}
    \norm*{\hat\bfmu_\ell-\bfmu_\ell}_2^2\leq \sigma^2
\end{eqn}
by setting $\delta=1$ in their paper. Set $t\coloneqq\max\{s,2m+k,d\}$. Then, we may extract the $t$ underlying points in $t^2 = O(n)$ kernel queries by reading a $t\times t$ submatrix of the kernel matrix and retrieving the underlying Gaussian points themselves from the inner product matrix up to a rotation, for instance by Cholesky decomposition. Since we have a sample of size at least $s$, we may approximate the Gaussian means. Now, of the $t$ Gaussian points sampled, we show that we can exactly assign which points belong to which Gaussians for $2m+k$ input points using Lemma \ref{lem:distinguish-gaussians}.

Now let $\bfx_1$ and $\bfx_2$ be two input points with the same mean. Then note that $\bfx_1-\bfx_2\sim\mathcal N(0,2\sigma^2\bfI_d)$ and that we may compute its inner product with another input point $\bfx'$ in two kernel queries, i.e.\ by computing $\bfx_1\cdot\bfx'$ and $\bfx_2\cdot\bfx'$ individually and subtracting them. Now let $\bfS\in\mathbb R^{m\times d}$ be the matrix formed by placing $m$ pairs of the above difference of pairs of Gaussians drawn from the same mean, scaled by $(2\sigma^2)^{-1}$. Then $\bfS$ is a $m\times d$ matrix of i.i.d.\ Gaussians, and for $n-2m$ input points $\bfx_i$, we may compute $\bfS\bfx_i$ with $O(nm)$ kernel queries total. We then prove that for well-separated Gaussian means, $\bfS\bfx_i$ can be used to identify which true Gaussian mean $\bfx_i$ came from in corollary \ref{cor:gaussian-sketch-dimension}.

Finally, we show that clustering points to their Gaussian means results in an approximately optimal $k$-means clustering. By the above, we can do this assigning for Gaussian means that are separated by more than $\eps\sigma^2 d$, and otherwise, assigning to a wrong mean only $\eps\sigma^2d$ away still results in a $(1+\eps)$-optimal clustering.

\subsection{Assigning input points to Gaussian means}
We first present the following lemma, which allows us to distinguish whether a point is drawn from a Gaussian with one mean or another with probability at least $1-\delta$.
\begin{lem}[Distinguishing Gaussian means]\label{lem:distinguish-gaussians}
Let $\bftheta_1,\bftheta_2\in\mathbb R^d$ be two Gaussian means separated by $\norm*{\bftheta_1-\bftheta_2}_2^2\geq C\sigma^2\log\delta^{-1}$ for a constant $C$ large enough and $\delta\in(0,1/2)$. Furthermore, let $\hat\bftheta_1,\hat\bftheta_2$ be approximations to the Gaussian means with $\norm*{\hat\bftheta_b - \bftheta_b}_2\leq\sigma$ for $b\in\{1,2\}$. Let $\hat\bfc\coloneqq(\hat\bftheta_1+\hat\bftheta_2)/2$. Then
\begin{eqn}
    \begin{cases}
        (\bfx-\hat\bfc)\cdot(\hat\bftheta_1-\hat\bfc) > 0 & \text{if $\bfx\sim\mathcal N(\bftheta_1,\sigma^2\bfI_d)$ w.p.\ at least $1-\delta$} \\
        (\bfx-\hat\bfc)\cdot(\hat\bftheta_1-\hat\bfc) < 0 & \text{if $\bfx\sim\mathcal N(\bftheta_2,\sigma^2\bfI_d)$ w.p.\ at least $1-\delta$}
    \end{cases}.
\end{eqn}
\end{lem}
\begin{proof}
Let $\bfx = \bftheta+\bfeta$ with $\bftheta\in\{\bftheta_1,\bftheta_2\}$ and $\bfeta\sim\mathcal N(0,\sigma^2\bfI_d)$. Then if $\bftheta = \bftheta_1$, then
\begin{eqn}\label{eqn:gaussian-test-theta-1}
    (\bfx-\hat\bfc)\cdot(\hat\bftheta_1-\hat\bfc) &= (\bftheta_1-\hat\bfc)\cdot(\hat\bftheta_1-\hat\bfc) + \bfeta\cdot(\hat\bftheta_1-\hat\bfc) \\
    &= \norm*{\hat\bftheta_1-\hat\bfc}_2^2 + (\bftheta_1-\hat\bftheta_1)\cdot(\hat\bftheta_1-\hat\bfc) + \bfeta\cdot(\hat\bftheta_1-\hat\bfc)
\end{eqn}
and similarly if $\bftheta = \bftheta_2$, then
\begin{eqn}\label{eqn:gaussian-test-theta-2}
    (\bfx-\hat\bfc)\cdot(\hat\bftheta_1-\hat\bfc) &= (\bftheta_2-\hat\bfc)\cdot(\hat\bftheta_1-\hat\bfc) + \bfeta\cdot(\hat\bftheta_1-\hat\bfc) \\
    &= (\hat\bftheta_2-\hat\bfc)\cdot(\hat\bftheta_1-\hat\bfc) + (\bftheta_2-\hat\bftheta_2)\cdot(\hat\bftheta_1-\hat\bfc) + \bfeta\cdot(\hat\bftheta_1-\hat\bfc).
\end{eqn}
Note that
\begin{eqn}
    \norm*{\frac{\hat\bftheta_1-\hat\bftheta_2}{2}-\frac{\bftheta_1-\bftheta_2}{2}}_2 \leq \frac12\parens*{\norm*{\bftheta_1-\hat\bftheta_1}_2+\norm*{\bftheta_2-\hat\bftheta_2}_2}\leq \sigma
\end{eqn}
and thus we have the following estimates:
\begin{eqn}
    \norm*{\hat\bftheta_1-\hat\bfc}_2^2 &\geq \norm*{\bftheta_1-\bfc}_2^2 - 2\sigma\norm*{\bftheta_1-\bfc}_2 \\
    (\hat\bftheta_1-\hat\bfc)\cdot(\hat\bftheta_2-\hat\bfc) &\leq -\norm*{\bftheta_1-\bfc}_2^2 + \sigma^2 + 2\sigma\norm*{\bftheta_1-\bfc}_2 \\
    \abs*{(\bftheta_1-\hat\bftheta_1)\cdot(\hat\bftheta_1-\hat\bfc)} &\leq \sigma^2 + \sigma\norm*{\bftheta_1-\bfc}_2 \\
    \abs*{(\bftheta_2-\hat\bftheta_2)\cdot(\hat\bftheta_1-\hat\bfc)} &\leq \sigma^2 + \sigma\norm*{\bftheta_1-\bfc}_2.
\end{eqn}
These bound all but the last terms in equations (\ref{eqn:gaussian-test-theta-1}) and (\ref{eqn:gaussian-test-theta-2}). To bound the last term, note also that we may take $\sigma\leq \norm*{\bftheta_1-\bftheta_2}_2/12$ for $C\geq 144/\log2$. Furthermore, $\bfeta\cdot(\hat\bftheta_1-\hat\bfc)\sim\mathcal N(0,\sigma^2\norm*{\hat\bftheta_1-\hat\bfc}_2^2)$ and thus with probability at least $1-\delta$, we have that
\begin{eqn}
    \abs*{\bfeta\cdot(\hat\bftheta_1-\hat\bfc)} &\leq \sigma\sqrt{\log \frac1\delta}\norm*{\hat\bftheta_1-\hat\bfc}_2\leq \frac1{6}\norm*{\bftheta_1-\bfc}\parens*{\norm*{\bftheta_1-\bfc}_2+\sigma}\leq \frac13\norm*{\bftheta_1-\bfc}_2^2
\end{eqn}
by taking $C\geq 6$. We also have the bounds
\begin{eqn}
    \sigma^2 &\leq \frac{\norm*{\bftheta_1-\bftheta_2}_2^2}{144} \\
    2\sigma\norm*{\bftheta_1-\bfc}_2 &\leq 2\sigma\frac{\norm*{\bftheta_1-\bfc}_2}{6}\norm*{\bftheta_1-\bfc}_2\leq \frac13\norm*{\bftheta_1-\bfc}_2^2.
\end{eqn}
Then, with probability at least $1-\delta$, if $\bftheta = \bftheta_1$, then
\begin{eqn}
    (\bfx-\hat\bfc)\cdot(\hat\bftheta_1-\hat\bfc) &\geq \norm*{\bftheta_1-\bfc}_2^2 - \frac13\norm*{\bftheta_1-\bfc}_2^2 > 0
\end{eqn}
and if $\bftheta = \bftheta_2$, then
\begin{eqn}
    (\bfx-\hat\bfc)\cdot(\hat\bftheta_1-\hat\bfc) &\leq -\norm*{\bftheta_1-\bfc}_2^2 + \frac{\norm*{\bftheta_1-\bfc}_2^2}{36} + \frac13\norm*{\bftheta_1-\bfc}_2^2 < 0
\end{eqn}
and thus we conclude as desired.
\end{proof}

Using Lemma \ref{lem:distinguish-gaussians}, we may identify the true Gaussian mean of a point with probability at least $1-(2m+k)^{-3}$ with squared separation only $O(\sigma^2(\log\log n + \log\eps^{-1} + \log k))$. Then by a union bound, we may indeed identify the true Gaussian means of $2m+k$ points simultaneously with high probability. We may thus form the matrix $\bfS$ of i.i.d.\ standard Gaussians as described previously and apply it to the $n-2m$ remaining points.

As a corollary of Lemma \ref{lem:distinguish-gaussians}, we show that for Gaussian means that are separated more, with squared distance at least $\eps\sigma^2 d$, we may distinguish the means with a Gaussian sketch of dimension $m = O(\eps^{-1}\log\delta^{-1})$ with probability at least $1-\delta$. In particular, we may choose the failure probability to be $\delta = (nk)^{-3}$ so that with a sketch dimension of $m = O(\eps^{-1}\log(nk)^3) = O(\eps^{-1}\log n)$, we can identify the correct Gaussian mean for all $n-2m$ remaining input points simultaneously by the union bound, as claimed. That is, using $\bfS\bfx$, we can find the correct mean of $\bfx$ for Gaussians with large enough separation.

\begin{cor}\label{cor:gaussian-sketch-dimension}
Let $\bfmu_1,\bfmu_2\in\mathbb R^d$ be two Gaussian means separated by $\norm*{\bfmu_1-\bfmu_2}_2^2\geq \eps\sigma^2 d$, and let $\delta\in(0,1/2)$. Let $\bfS\in\mathbb R^{m\times d}$ be a matrix of i.i.d.\ standard Gaussians. If $m \geq C\eps^{-1}\log(\delta^{-1})$, for some constant $C$ large enough, then there exists an algorithm that can decide whether $\bfx\sim\mathcal N(\bfmu_1,\sigma^2\bfI_d)$ or $\bfx\sim\mathcal N(\bfmu_2,\sigma^2\bfI_d)$ given only $\bfS$, $\bfS\bfx$, and the approximate means $\hat\bfmu_j$, with probability at least $1-\delta$.
\end{cor}
\begin{proof}
Let $\bfS = \bfU\bfSigma\bfV^\top$ be the truncated SVD of $\bfS^\top$. Note that the algorithm can compute this decomposition and thus can retrieve $\bfV^\top \bfx = (\bfU\bfSigma)^{-1}\bfS\bfx$ and that $\bfV^\top$ is a random projection. Then as discussed in the proof of Theorem 2.1 in \cite{dasgupta2003elementary}, we have
\begin{eqn}
    \E_{\bfS}\parens*{\norm*{\bfV^\top(\bfmu_1-\bfmu_2)}_2^2} = \frac{m}{d}\norm*{\bfmu_1-\bfmu_2}_2^2,
\end{eqn}
and by Lemma 2.2 of \cite{dasgupta2003elementary},
\begin{eqn}
    \Pr_{\bfS}\parens*{\norm*{\bfV^\top(\bfmu_1-\bfmu_2)}_2^2\leq \frac12\frac{m}{d}\norm*{\bfmu_1-\bfmu_2}_2^2} < \exp\parens*{-\Omega(m)}\leq \frac{\delta}{2}
\end{eqn}
for $C$ chosen large enough. Now suppose that the above event does not happen, which happens with probability at least $1-\delta/2$. Write $\bfx = \bfmu + \bfeta$, where $\bfmu\in\{\bfmu_1,\bfmu_2\}$ and $\bfeta\sim\mathcal N(0,\sigma^2\bfI_d)$. Note then that $\bfV^\top\bfx = \bfV^\top\bfmu + \bfV^\top\bfeta\sim\mathcal N(\bfV^\top\bfmu,\sigma^2\bfI_m)$ by the rotational invariance of Gaussians, and
\begin{eqn}
    \norm*{\bfV^\top\bfmu_1-\bfV^\top\bfmu_2}_2^2\geq \frac1{2}\frac{m}{d}\norm*{\bfmu_1-\bfmu_2}_2^2\geq \frac1{2}\frac{m}{d}\eps\sigma^2d\geq \frac{C}2\log\delta^{-1}.
\end{eqn}
Furthermore, we have an approximation of the means $\bfV^\top\bfmu_1$ and $\bfV^\top\bfmu_2$ via $\bfV^\top\hat\bfmu_1$ and $\bfV^\top\hat\bfmu_2$ with
\begin{eqn}
    \norm*{\bfV^\top(\bfmu_b-\hat\bfmu_b)}_2\leq \norm*{\bfmu_b-\hat\bfmu_b}_2 = \sigma^2.
\end{eqn}
We then take our $C$ here to be big enough to use Lemma \ref{lem:distinguish-gaussians} and conclude.
\end{proof}
We now put corollary \ref{cor:gaussian-sketch-dimension} to algorithmic use by using it to assign to each point a center withing $\eps\sigma^2 d$.

\begin{lem}\label{lem:assign-gaussians}
    With probability at least $99/100$, we may simultaneously assign for each $\bfx_i$ for $i\in[n]$ a center $\bfmu_{\ell_i}$ with $\norm*{\bfmu_{\ell_i}-\bfmu_{\ell_i^*}}_2^2\leq\eps\sigma^2 d$, where $\bfmu_{\ell_i^*}$ is the true Gaussian mean that generated $\bfx_i$. Furthermore, the assignment algorithm that we describe only depends on $\bfS$, $\bfS\bfx_i$, and approximate means $\hat\bfmu_j$.
\end{lem}
\begin{proof}
We claim that we can assign a center within squared distance $\eps\sigma^2 d$ as follows. Let $\bfx = \bfmu+\bfeta$ with $\bfeta\sim\mathcal N(0,\sigma^2\bfI_d)$. We then iterate through guesses $\bfmu_j$ for $j\in[k]$ and assign $\bfmu_j$ to $\bfx$ if we run the hypothesis test between $\bfx\sim\mathcal N(\bfmu_j,\sigma^2\bfI_d)$ and $\bfx\sim\mathcal N(\bfmu_\ell,\sigma^2\bfI_d)$ and $\bfx$ chooses $\bfmu_j$ for \emph{every} $\ell\in[k]\setminus\{j\}$. Recall that we set the failure rate $\delta$ in corollary \ref{cor:gaussian-sketch-dimension} to $(nk)^{-3}$, so the hypothesis test is accurate for all $nk(k-1)$ hypothesis tests ranging over $n$ data points, $j\in[k]$, and $\ell\in[k]\setminus\{\ell\}$. Clearly, guessing $\bfmu_j = \bfmu$ results in passing all the hypothesis tests in this case. Note that our hypothesis test is run using corollary \ref{cor:gaussian-sketch-dimension} and only depends on $\bfS$, $\bfS\bfx_i$, and approximate means $\hat\bfmu_j$.

Now suppose that $\bfmu_j$ is a center that is more than $\eps\sigma^2d$ squared distance away from $\bfmu$. Then when we guess $\bfmu_j$, $\bfmu_j$ fails at least one hypothesis test, namely the one testing $\mathcal N(\bfmu_j,\sigma^2\bfI_d)$ against $\mathcal N(\bfmu,\sigma^2\bfI_d)$ when $\bfmu_\ell = \bfmu$. Thus, this algorithm correctly assigns every Gaussian input point to a center that is at most $\eps\sigma^2 d$ square distance away from the true mean $\bfmu$.
\end{proof}

\subsection{Clustering the points}
Now that we have approximately assigned input points to Gaussian means in $O(nm) = \tilde O(n/\eps)$ kernel queries, it remains to show that this information suffices to give a $(1+\eps)$-approximate solution to the KKMC problem.

\begin{thm}\label{thm:cluster-gaussian-same}
Let $d\eps\geq 1$ and $k/\eps\leq d\leq n/10$ and let our data set $\{\bfx_i\}_{i=1}^n$ be distributed as a mixture of $k$ Gaussians as described before. Then assigning the $\bfx_i$ to approximate means as in Lemma \ref{lem:assign-gaussians} gives a $(1+8\eps)$-approximate $k$-means clustering solution with probability at least $98/100$.
\end{thm}
\begin{proof}
Let $\bfX\in\mathbb R^{n\times d}$ be the design matrix of points with our dataset drawn from a mixture of Gaussians in the rows. Now write $\bfX = \bfM + \bfG$, where $\bfM$ is the matrix with the Gaussian mean of each point in the rows and $\bfG$ is a matrix with rows all distributed as $\mathcal N(0,\sigma^2\bfI_d)$.

\paragraph{Lower bounds on the cost.}
We first bound below the cost of any $k$-means clustering solution, i.e.\ an assignment of $k$ centers to each of the $n$ points. We may then place these centers in the rows of a matrix $\bfC$, so that the input data point $\bfx_i = \bfe_i^\top\bfX$ is assigned the $k$-means center $\bfe_i^\top\bfC$ for $i\in[n]$. The cost of this $k$-means solution is then
\begin{eqn}
    \norm*{\bfX - \bfC}_F^2 = \norm*{\bfG + \bfM - \bfC}_F^2.
\end{eqn}
Now note that $\bfM-\bfC$ has rank at most $2k$, so the above cost is bounded below by the cost of the best rank $2k$ approximation of $\bfG$ in Frobenius norm. Furthermore, by the Eckart-Young-Mirky theorem, the cost of the optimal low rank approximation is the sum of the smallest $d-2k$ squared singular values of $\bfG$.

Let $s_1(\bfG)\geq s_2(\bfG)\geq\dots\geq s_d(\bfG)$ denote the singular values of $\bfG$. Note that $\bfG/\sigma$ is a matrix with i.i.d.\ standard Gaussians, so by results summarized in \cite{rudelson2009smallest}, we have that $(1/2)\sigma\sqrt{n}\leq s_n(\bfG)\leq s_1(\bfG)\leq (3/2)\sigma\sqrt{n}$ with probability $99/100$ for $n$ large enough and $d\leq n/10$. Then,
\begin{eqn}
    \sum_{i=1}^{2k}s_i(\bfG)^2\leq (2k)s_1(\bfG)^2\leq 3(2k)s_n(\bfG)^2\leq \frac{6k}d\sum_{i=1}^d s_i(\bfG)^2
\end{eqn}
and thus
\begin{eqn}\label{eqn:gaussian-cost-lower-bound}
    \norm*{\bfX-\bfC}_F^2\geq \sum_{i=1}^{d-2k}s_i(\bfG)^2\geq \parens*{1-\frac{6k}{d}}\sum_{i=1}^d s_i(\bfG)^2 = \parens*{1-\frac{6k}{d}}\norm*{\bfG}_F^2\geq \parens*{1-6\eps}\norm*{\bfG}_F^2.
\end{eqn}

\paragraph{The cost of clustering by the Gaussian means.}
We now give an algorithm using the approximate Gaussian means and our Gaussian mean assignment algorithms. Note that if we can correctly cluster every input point to its Gaussian center, then the resulting clustering has cost at most $\norm*{\bfG}_F^2$ since using the empirical center of the Gaussians will have lower cost than the true means. Then for $d\geq k/\eps$, we have that this clustering has cost at most $1/(1-4\eps)\leq 1+5\eps$ times the optimal clustering by equation (\ref{eqn:gaussian-cost-lower-bound}).

However, with our kernel query budget, we can only do this assignment for Gaussian means that are separated by squared distance $\eps\sigma^2 d$ using Lemma \ref{lem:distinguish-gaussians}; for separation smaller than this, we cannot disambiguate. The fix here is that we in fact do not need to, since assigning to a Gaussian mean that is less than $\eps\sigma^2 d$ does not change the cost by more than a $(1+\eps)$ for that point.

Now consider an input point $\bfx = \bfmu^*+\bfeta$ with $\bfeta\sim\mathcal N(0,\sigma^2\bfI_d)$, let $\bfmu$ be the true mean assigned to $\bfx$ in Lemma \ref{lem:assign-gaussians}, and let $\hat\bfmu$ be the approximation that approximates $\bfmu$. We then have that
\begin{eqn}
    \norm*{\bfx-\hat\bfmu}_2^2 &= \norm*{\bfeta+(\bfmu^*-\hat\bfmu)}_2^2 = \norm*{\bfeta}_2^2 + \norm*{\bfmu_*-\hat\bfmu}_2^2 + 2\angle*{\bfeta,\bfmu^*-\hat\bfmu}.
\end{eqn}
Note that $\norm*{\bfmu_*-\hat\bfmu}_2^2 \leq \norm*{(\bfmu_*-\bfmu)+(\bfmu-\hat\bfmu)}_2^2\leq 2\eps\sigma^2 d + 2\sigma^2$ and that $\angle*{\bfeta,\bfmu^*-\hat\bfmu}\sim\mathcal N(0,\sigma^2\norm*{\bfmu^*-\hat\bfmu}_2^2)$. Thus, we have that
\begin{eqn}
    \norm*{\bfx-\hat\bfmu}_2^2 &= \norm*{\bfeta}_2^2 + 2\eps\sigma^2 d+2\sigma^2 + 2\angle*{\bfeta,\bfmu^*-\hat\bfmu}\leq \norm*{\bfeta}_2^2 + 4\eps\sigma^2 d + 2\angle*{\bfeta,\bfmu^*-\hat\bfmu}.
\end{eqn}
Note that summing the $2\angle*{\bfeta,\bfmu^*-\hat\bfmu}$ term over $n$ input points gives a zero mean Gaussian with variance at most $2n\sigma^2(2\eps\sigma^2 d + 2\sigma^2)\leq 4n\sigma^4\eps d$. With probability $99/100$, this is bounded by $4\sigma^2\sqrt{n\eps d}\leq \eps n\sigma^2 d$ for $n$ large enough.

Thus, we then have that
\begin{eqn}
    \sum_{i=1}^n \norm*{\bfx_i-\hat\bfmu}_2^2\leq \norm*{\bfG}_F^2 + 5\eps(n\sigma^2 d).
\end{eqn}
Now note that $\norm*{\bfG}_F^2/\sigma^2$ is a chi-squared variable with $nd$ degrees of freedom, so by concentration bounds found in \cite{laurent2000adaptive}, we have that
\begin{eqn}
    \Pr\parens*{\norm*{\bfG}_F^2/\sigma^2-nd\leq 2nd}\leq \exp(-nd)
\end{eqn}
and thus with probability at least $99/100$ for $nd$ large enough,
\begin{eqn}
    5\eps(n\sigma^2 d)\leq \frac53\eps\norm*{\bfG}_F^2.
\end{eqn}
We thus conclude that with probability at least $98/100$, the above algorithm gives an approximation ratio of at most
\begin{eqn}
    \frac{1+5\eps/3}{1-6\eps}\leq 1+8\eps
\end{eqn}
as claimed.
\end{proof}

\section{Acknowledgements}
We would like to thank the anonymous reviewers for helpful feedback. D. Woodruff would like to thank support from the Office of Naval Research (ONR) grant N00014-18-1-2562. This work was also partly done while D. Woodruff was visiting the Simons Institute for the Theory of Computing.

\bibliographystyle{alpha}
\bibliography{citations}

\appendix

\end{document}